\documentclass[12pt]{article}

\usepackage[margin=0.8in]{geometry}

\usepackage[usenames,dvipsnames]{xcolor}
\usepackage{amssymb,amsmath,xspace,enumerate,mathrsfs,algpseudocode,booktabs,adjustbox,soul,threeparttable,afterpage,longtable,comment,multirow,etoolbox,subcaption,enumerate,longtable,lscape, tabularx,cellspace,svg}

\usepackage{graphicx, pgfplots, standalone, pgf, tikz}
\usepackage[justification=centering]{caption}

\usepackage[sort&compress,numbers]{natbib}

\usepackage{comment, url,float, setspace, xspace}
\pgfplotsset{compat=1.17} 

\usepackage[inline]{enumitem}
\usepackage{amsmath, amssymb, amsthm, amsfonts, mathtools, array, bm, mathrsfs}
\renewcommand\arraystretch{1.2}
\setlist[enumerate]{label=(\roman*.)}

\usepackage[ruled,vlined,linesnumbered]{algorithm2e}

\usepackage[short,c2,nocomma]{optidef}

\usepackage[colorlinks,
            linkcolor=BlueViolet,
            citecolor=blue,
            urlcolor=blue]{hyperref}

\usepackage[nameinlink]{cleveref}

\usepackage{csvsimple, booktabs, longtable, multirow, multicol, adjustbox}	%

\usepackage[xindy, acronyms, nowarn]{glossaries}   %
\newtheorem{remark}{Remark}
\newtheorem{example}{Example}
\newtheorem{definition}{Definition}
\newtheorem{theorem}{Theorem}
\newtheorem{corollary}{Corollary}
\newtheorem{lemma}{Lemma}
\newtheorem{claim}{Claim}

\definecolor{lime}{HTML}{A6CE39}
\DeclareRobustCommand{\orcidicon}{
	\begin{tikzpicture} \draw[lime, fill=lime] (0,0) circle [radius=0.16] node[white] { {\fontfamily{qag}\selectfont \tiny ID} };
	\draw[white, fill=white] (-0.0625,0.095) circle [radius=0.007];
	\end{tikzpicture} \hspace{-2mm}
}

\foreach \x in {A, ..., Z}{\expandafter\xdef\csname orcid\x\endcsname{\noexpand\href{https://orcid.org/\csname orcidauthor\x\endcsname}{\noexpand\orcidicon} } }
\renewcommand{\bar}{\overline}
\renewcommand{\tilde}{\widetilde}
\renewcommand{\hat}{\widehat}
\colorlet{rosso}{red!90!white}
\definecolor{green}{RGB}{45, 204, 53}
\definecolor{blue}{RGB}{59, 42, 247}

\allowdisplaybreaks
\crefname{lemma}{Lemma}{Lemmata}
\crefname{theorem}{Theorem}{Theorems}
\crefname{claim}{Claim}{Claims}
\crefname{algorithm}{Algorithm}{Algorithms}
\crefname{equation}{}{}
\crefname{definition}{Definition}{Definition}
\crefname{Cla}{Claim}{Claim}
\crefname{corollary}{Corollary}{Corollaries}
\crefname{remark}{Remark}{Remarks}
\crefname{example}{Example}{Examples}
\crefname{figure}{Figure}{Figures}
\crefname{section}{Section}{Sections}
\crefname{table}{Table}{Tables}
\crefname{enumi}{Statement}{Statements}
\crefname{line}{Step}{Steps}

\newcommand{\SSI}{\texttt{SSI}\xspace}
\newcommand{\PNE}{\emph{PNE}\xspace}
\newcommand{\PNEs}{\emph{PNE}s\xspace}

\newcommand{\MNE}{\emph{MNE}\xspace}
\newcommand{\MNEs}{\emph{MNE}s\xspace}
\newcommand{\NASP}{\emph{NASP}\xspace}
\newcommand{\NASPs}{\emph{NASP}s\xspace}
\newcommand{\LCP}{\emph{LCP}\xspace}

\newcommand{\EPEC}{\emph{EPEC}\xspace}
\newcommand{\EPECs}{\emph{EPEC}s\xspace}
\newcommand{\MIP}{\emph{MIP}\xspace}

\newcommand{\NPC}{$\mathcal{NP}$-complete\xspace}

\newcommand{\SigmaTwoP}{$\Sigma^p_2$-hard\xspace}

\newcommand{\q}		{\mathbf{\textcolor{blue}{q}}}
\newcommand{\Cost}	{\mathbf{\textcolor{rosso}{C}}}
\newcommand{\CostQ}	{\mathbf{\textcolor{rosso}{D}}}

\newcommand{\taxb}[1][r]  {\textcolor{rosso}{b^{#1}}}
\newcommand{\Cemm}[1][p] {\Cost^{#1}_{\text{emmision}}}
\newcommand{\qi}[1][p] 	{\q^{#1}}
\newcommand{\qiCap}[1][p] 	{\overline{\mathbf{\textcolor{rosso}{q^{#1}}}}}
\newcommand{\ti}[1][p] 	{\mathbf{\textcolor{blue}{t}}^{#1}}
\newcommand{\tiCap}[1][p] 	{\overline{\mathbf{\textcolor{rosso}{t}}^{#1}}}

\newcommand{\piCeil}[1][C] 	{\overline{\mathbf{\textcolor{rosso}{\pi}}^{#1}}}
\newcommand{\DemInt}[1][C] {\textcolor{rosso}{\alpha^{#1}}}
\newcommand{\DemSlope}[1][C] {\textcolor{rosso}{\beta^{#1}}}
\newcommand{\qimp}[1][C] {\q^{#1}_{\text{imp}}}

\newcommand{\qexp}[1][C] {\q^{#1}_{\text{exp}}}
\newcommand{\Cilin}[1][p]	{\Cost^{#1}}
\newcommand{\Ciquad}[1][p]	{\CostQ^{#1}}

\newcommand{\qAimpI}[1][I\to A] {\q^{#1}_{\text{imp}}}

\newcommand{\piI}[1][I] {\mathbf{\textcolor{green}{\pi}}^{#1}}

\newcommand{\R}{\mathbb{R}}

\newcommand{\Z}{\mathbb{Z}}
\newcommand{\vecN}[1]{ \left(\begin{array}{c}#1\end{array}\right) }

\newcommand{\sol}{\operatorname{SOL}}
\newcommand{\conv}{\operatorname{conv}}

\newcommand{\cl}{\operatorname{cl}}

\title{When Nash Meets Stackelberg}
\author{ Margarida Carvalho\orcidA{} Gabriele Dragotto\orcidB{} Felipe Feijoo\orcidC{}\\ Andrea Lodi\orcidD{} Sriram Sankaranarayanan\orcidE{} }

\begin{document}
\maketitle

\begin{abstract}This article introduces a class of \emph{Nash} games among \emph{Stackelberg} players (\NASPs), namely, a class of simultaneous non-cooperative games where the players solve sequential Stackelberg games. Specifically, each player solves a Stackelberg game where a leader optimizes a (parametrized) linear objective function subject to linear constraints while its followers solve convex quadratic problems subject to the standard optimistic assumption. 
Although we prove that deciding if a \NASP instance admits a Nash equilibrium is generally a \SigmaTwoP decision problem, we devise two exact and computationally-efficient algorithms to compute and select Nash equilibria or certify that no equilibrium exists. 
We employ \NASPs to model the hierarchical interactions of international energy markets where climate-change aware regulators oversee the operations of profit-driven energy producers. By combining real-world data with our models, we find that Nash equilibria provide informative, and often counterintuitive, managerial insights for market regulators.
\end{abstract}

\section{Introduction} \label{sec:intro}

During the past decades, sustained methodological and practical advances in optimization promoted the development of reliable and efficient technologies to solve complex optimization problems. As a result, companies, organizations, and decision-makers often solve optimization problems to plan their operations and improve their efficiency. However, decision-making is rarely an individual task; on the contrary, it comprises several interconnected and self-driven decision-makers. To address the interactive and complex nature of decision-making, we arguably need novel frameworks combining the modeling capabilities of optimization with game theory.
Despite this promising perspective, game-theoretical frameworks are as helpful as our ability to design efficient algorithms to compute their outcomes, for instance, Nash (i.e., simultaneous) and Stackelberg (i.e., sequential) equilibria.
Motivated by these challenges, we provide models, algorithms, and theory concerning the Nash equilibria of games where self-driven players solve optimization problems.

\begin{figure}[t]
	\centering
	\includegraphics[width=0.85\textwidth]{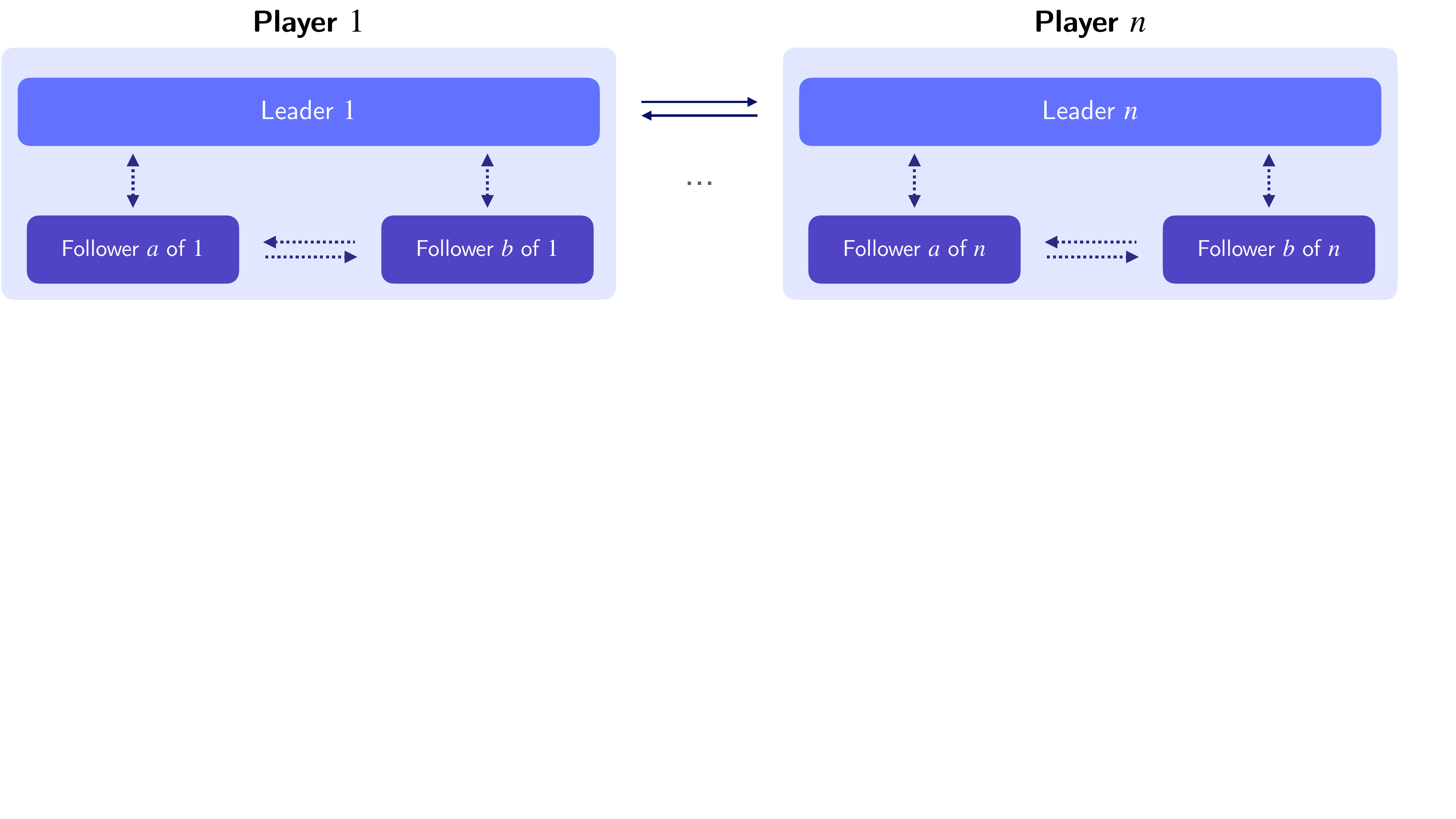}
	\caption{A schematic representation of a \NASP with $n$ players with $2$ followers each. The vertical arrows represent sequential “Stackelberg” interactions, while the horizontal ones are simultaneous “Nash” interactions. }\label{fig:Schematic}
\end{figure}

\paragraph{The Game.}  This paper introduces a class of games capturing a hierarchical system of sequential and simultaneous interactions among decision-makers as shown in \cref{fig:Schematic}. Specifically, we study a class of \emph{complete-information}, \emph{simultaneous} and \emph{non-cooperative} “Nash” games among “Stackelberg” players (\NASPs) where:
\begin{enumerate}[label=(\roman*)]
	\item There are $n$ first-mover players called the {\em leaders}, each of which has a set of lower-level agents called the {\em followers}. Leaders and followers hierarchically decide their strategies by solving optimization problems, and the parameters of such problems are common information.
	\item Each follower is associated with a unique leader and decides by optimizing a parametrized convex quadratic function over a set of linear constraints. Precisely, each follower solves an optimization problem where the objective function and the feasible region are parametrized in the decision variables of its leader. Furthermore, each follower simultaneously interacts with the other followers associated with the same leader; therefore, its objective function is also parametrized in the other followers' variables. 
	\item Each leader optimizes an objective function that is linear in its variables and parametrized in the other leaders' variables through bilinear terms. Furthermore, each leader's optimization problem is constrained by linear inequalities involving the leader's and its followers' variables.
	\item The leaders act simultaneously and non-cooperatively at the upper level. Given the decisions of each leader, its followers play simultaneously and non-cooperatively at the lower level. Specifically, each follower interacts with the followers having the same leader and does not interact with the other leaders' followers.
\end{enumerate}
Whenever we embed the followers' optimization problems inside their associated leader's optimization problem, we obtain a so-called \emph{Stackelberg game}.
In other words, a Stackelberg game is an optimization problem maximizing the leader's objective subject to its constraints \emph{and} the optimality of its followers' optimization problems. 
Since the latter can be equivalently formulated through each follower's Karush–Kuhn–Tucker conditions ($KKT$), we can express each Stackelberg game as a non-convex optimization problem that embeds the hierarchical interactions among a leader and its followers.
As a consequence, we define a \NASP as a simultaneous game where each \emph{player} solves a Stackelberg game. 

\begin{remark}
We refer to the Stackelberg games associated with each leader as the \emph{players}.
The feasible set of a leader is defined by the linear inequalities directly associated with the leader's optimization problem. In contrast, the feasible set of a player is defined by the leader's constraints \emph{and} the optimality of the followers' optimization problems (or, equivalently, the $KKT$ conditions).
\end{remark}

\NASPs fall in the category of \emph{Equilibrium Problems with Equilibrium Constraints} (\EPECs), a class of problems possessing several applications in energy and pricing contexts \citep{Gabriel2012,Hu2007,Ralph,Pozo2011}. The Stackelberg leaders may represent first-mover decision-makers, for instance, governmental agencies and regulatory bodies; their followers may represent market players simultaneously optimizing their benefits under their leader's regulations.
In \cref{ex:first_example}, we showcase a simple \NASP.

\begin{example}
	Consider a \NASP with $2$ players, the {\em Latin} player and the {\em Greek} player, having $2$ followers and $1$ follower, respectively. Let $w$ be the variables of the Latin leader and $y^1$ and $y^2$ be the variables of its followers. Let $\xi$ be the Greek leader's variables and $\chi$ be the variables of its follower. Given $\xi,\chi$ as parameters, the Latin player solves the Stackelberg game
	\begin{subequations}
		\begin{align}
			\min_{w,y} \quad\quad    & c^\top \vecN{w                                                                                                                         \\y} + \left( C\vecN{\xi                                                     \\\chi}
			\right)^\top  \vecN{w\\y}   \label{eq:Lat:Obj} \\
			\text{s.t.} \qquad & Aw + B\vecN{y^1                                                                                                                        \\ y^2}\quad\leq\quad b,   \\
			                         & y^1 \quad\in\quad \arg\min_{y^1} \left\{\frac{1}{2}(y^1)^\top F y^1 + f^\top y^1 + (Gy^2 + Hw)^\top y^1: Lw + Ny^1  \leq d  \right\},  \\
			                         & y^2 \quad\in\quad \arg\min_{y^2} \left\{\frac{1}{2} (y^2)^\top O y^2 + h^\top y^2 + (Py^1 + Qw)^\top y^2: Rw + Sy^2  \leq e  \right\}. 
		\end{align}
		\noindent Similarly, given $w,y=(y^1,y^2)$ as parameters, the Greek player solves the Stackelberg game
		\begin{align}
			\min_{\xi,\chi}\quad\quad &                                                                                                                                                                    
			\alpha^\top \vecN{\xi\\\chi} + \left( \Gamma\vecN{w\\y} 
			\right)^\top \vecN{\xi\\\chi}  \qquad\qquad\qquad\qquad\qquad\qquad\qquad\quad \label{eq:Gr:Obj}  \\
			\text{s.t.} \qquad  & \Phi \xi + \Psi \chi \quad\leq\quad \beta,                                                                                                                         
			\\
			                          & \chi \quad\in\quad \arg\min_\chi \left\{ \frac{1}{2}\chi^\top \Delta \chi +\Phi^\top \chi + (\Upsilon \xi)^\top \chi  :  \Pi\xi + \Omega\chi \leq \gamma \right\}. 
		\end{align}
		\label{eq:LatinGreek}
	\end{subequations}
	In the above \NASP formulation, $F$, $O$ and $\Delta$ are symmetric positive semi-definite matrices, $c$, $b$, $d$, $e$, $f$, $h$, $\alpha$, $\beta$, $\gamma$ are parameter vectors of appropriate dimensions. The remaining objects are matrices of appropriate dimensions. The objective function of each player is parametrized with respect to the variables of the other players. The objective function of each follower is parametrized in (i.) its leader's variables, and (ii.) the variables of the other followers associated with its leader.
	\label{ex:first_example}
\end{example}

\paragraph{Applications. }
\NASPs provide a flexible modeling framework for many economic markets where several regulatory bodies interact and hierarchically regulate a set of lower-level economic agents. We outline three different applications related to energy, vaccine production, and insurance markets. 
First, the framework of \NASPs is mainly motivated by international energy markets with climate change-aware regulatory authorities and profit-maximizing energy producers. In this context, each energy producer, i.e., each follower, competes in its respective domestic market regulated by a single regulatory agency through taxes (e.g., a carbon tax) and production caps. Each regulatory agency, i.e., each leader, negotiates environmental agreements for trading energy and interacts with other regulatory authorities. As we will show, the \NASP abstraction practically captures the structure of such energy markets and it provides a general framework to derive practical insights into the effectiveness of environmental and regulatory initiatives. Recently, \citet{anjos_multinational_2022} studied a complex multinational carbon-credit system for the energy market exploiting the models and the algorithms available in this work.
Second, \NASPs can model complex drug trade markets. For instance, the COVID-19 vaccine production and exchange posed a severe threat to the world's immunization programs during the pandemic. Some experts warned of \emph{vaccine nationalism} \citep{vacnationalism_2021}, a phenomenon encompassing the strict export and import regulations several countries have imposed on vaccines and the raw materials required for their production \citep{fleming_2021, boffey_2021, boffey2_2021}.
In this context, countries would act as leaders by regulating the trade of vaccines and providing incentives to domestic vaccine producers that would serve as followers. The leaders' objective functions could model several tactical requirements; for instance, they can prioritize the production of doses reserved for vulnerable classes of the population or incentivize  exports to neighboring countries.
Third, \NASPs can model several dynamics associated with insurance markets, specifically, hierarchical multi-insurer games embedding \emph{re-insurance} mechanisms \citep{gordon2002economics,cashell2004economic}. Companies acting as followers contract insurance services to protect their operations from disruptions, e.g., cyberattacks. The insurers, acting as leaders, sell insurance products to their followers. The insurers may also mutually protect their portfolios to avoid losses due to large-scale disruptions, for instance, natural disasters.

\subsection*{Primary Contributions} 
We introduce \NASPs, a class of games encompassing a series of hierarchical and simultaneous interactions among several decision-makers solving optimization problems. While the solution concept for each Stackelberg game is the Stackelberg equilibrium, namely, a solution where the followers' decisions are optimal given the leader's decisions, we employ the Nash equilibrium as the standard solution concept for \NASPs. In a Nash equilibrium, players are mutually optimal and cannot unilaterally deviate from the equilibrium without diminishing their benefits, i.e., decreasing their objective function value.
We distinguish between Pure-Strategy Nash equilibria (\PNEs) and Mixed-Strategy Nash Equilibria (\MNEs); players employ deterministic strategies in the former, while in the latter, players randomize over their strategies. Unless otherwise stated, we focus on the more general concept of \MNE. We provide several theoretical, algorithmic, and practical contributions concerning the existence, computation, and interpretation of Nash equilibria in \NASPs. Precisely:

\begin{enumerate}[label=(\roman*)]
	\item We characterize the hardness of deciding whether a given instance of \NASP admits a Nash equilibrium or not as a \SigmaTwoP decision problem (\cref{sec:Hardness}). In other words, as long as $\mathcal{N}\mathcal{P} \neq \Sigma_2^p$, it is impossible to represent this problem as an integer program of polynomial size~\citep{Woeginger2021}.
	Under some mild conditions, we show that some \NASP instances always admit an \MNE (\cref{thm:BndObv}).
	      	      	      
	\item We provide the first exact and computationally-efficient algorithm to compute and select Nash equilibria in \NASPs (\cref{sec:Algo}). In contrast to the previous literature, our algorithm is exact and terminates with either a Nash equilibrium or a proof of its non-existence. We introduce two additional variants of our algorithm: a variant built upon an iterative inner-approximation scheme and a variant computing \PNEs (\cref{sec:enhance}). Our algorithms are based on the insight that any \NASP possesses an equivalent \emph{convex} representation (\cref{thm:Alg,rem:convHullisEnuf}).
	      	      	      
	\item We apply \NASPs to model international energy markets where climate-change aware regulators oversee the operations of profit-driven energy producers (\cref{sec:Compute}). We provide a detailed computational analysis of the performance of our algorithms on a set of synthetic energy-market instances. Furthermore, by combining real-world data and our models, we derive informative yet counterintuitive managerial insights from the \MNEs. We also provide an analysis of the Chilean-Argentinean energy market and derive some insights for the policymakers.
	      	      	      
\end{enumerate}

We organize the manuscript as follows. In \cref{sec:literature}, we provide a literature review, while in \cref{sec:Definitions}, we introduce the background definitions. In \cref{sec:Hardness},  we present an overview of the computational complexity results. In \cref{sec:Algo}, we present and characterize the algorithm to find equilibria in \NASPs. In \cref{sec:enhance}, we introduce an inner approximation algorithm and an algorithm to exclusively compute \PNEs. In \cref{sec:Compute}, we present the energy model, the computational tests, and the managerial insights. Finally, in \cref{sec:conclusions}, we provide some concluding remarks.

\section{Literature Review} 
\label{sec:literature}
\citet{Nash1951, Nash1950} introduced the concept of Nash Equilibrium in the context of finite games, i.e., games with a finite number of players and strategies. Nash proved that if the game is finite, there always exists an \MNE.
Within the optimization community, games expressed through the parametrized optimization problems associated with the players are often known as Nash equilibrium problems or Nash games. The Nash equilibrium concept has dramatically changed several scientific fields due to its flexibility and interpretability. 
Several authors employed Nash equilibria to capture the outcome of structured interactions of players solving optimization problems. For instance, gas market modeling \citep{egging2010world,feijoo2016north,Sankaranarayanan2018,feijoo2018future,holz2008strategic,egging2008complementarity,Stein2018}, cross-border kidney exchange models \citep{Carvalho2017c,Carvalho2022kidney}, competitive lot-sizing models \citep{Li2011,Carvalho2018}, knapsack and network-formation games \citep{Dragotto_2021_zero,Carvalho_2021_cut}, and fixed-charge transportation models \citep{Merx2018} employed the concept of Nash equilibrium.

In contrast to simultaneous games, sequential games partition the set of players into different groups, each playing in a predetermined round.
When there are two rounds, the game is known as a Stackelberg game \citep{Stackelberg1934Original}, with the first-round players being the leaders and the second-round players being the followers. In the optimization literature, Stackelberg games where players solve optimization problems are related to bilevel programming \citep{colson_bilevel_2005}, and their applications span several domains. For instance, \citet{Bard1998,Bard2000} modeled taxation strategies in the context of biofuel production, \citet{Brotcorne2008,Labbe2013,grimm_optimal_2021} modeled bilevel pricing problems, and \citet{hobbs2000strategic,gabriel2010solving,feijoo2014design} modeled pricing and environmental policies for energy markets, with power generators being leaders and network operators being followers.

\paragraph{Algorithms. } When multiple Stackelberg leaders, each possibly having multiple followers, interact, the game belongs to the family of \EPECs. Their application in economics often involves a multi-leader, multi-follower game \citep{pang_quasi-variational_2005}.
Several authors employed \EPECs to represent economic markets involving several decision-makers. \citet{sherali1984} introduced \EPECs where both leaders and followers produce a homogeneous commodity, and \citet{Ralph} and \citet{Hu2007} studied the existence of \PNEs in some specialized classes of \EPECs arising in electricity markets. \citet{demiguel2009stochastic} crafted the concept of stochastic multi-leader Stackelberg-Nash-Cournot equilibrium for a particular form of investment-production interactions.
Regarding methodological contributions, \citet{Gabriel2012} provided an iterative algorithm to compute \PNEs in a restricted class of \EPECs where followers from distinct leaders interact. \citet{leyffer2010solving} introduced an alternative solution concept based on a nonlinear programming reformulation.  \citet{Kulkarni2014S,Kulkarni2015e} analyzed \EPECs with shared constraints and introduced solution concepts and algorithms when the players' objective functions fulfill specific properties. Recently, \citet{devine_strategic_2022} introduced an \EPEC model for electricity markets with price-making firms with market power and price-taking firms. 
In contrast to this work, the previous works on \EPECs either: (i.) proposed algorithms to exclusively compute \PNEs, or (ii.) focused on weaker notions of equilibria. To the best of our knowledge, this is the first work providing an algorithm to compute exact \MNEs for the large class of \EPECs that \NASPs represent.

\paragraph{Complexity and Equilibria. } The two paramount issues concerning Nash equilibria are existence, namely, determining when at least one equilibrium exists, and computation, namely, devising efficient algorithms to compute equilibria. The original proof of existence from \citet{Nash1951, Nash1950} holds only for finite games, is non-constructive, and provides no methodology for computing or analytically constructing equilibria. Indeed, even if an equilibrium always exists, as in finite games or games with specific structures \citep{Rosenthal1973,Pia2017}, the problem of computing it is often not trivial, even in simple $2$-player cases \citep{Chen2006}. Furthermore, many variations of the decision version of the equilibrium problem are \NPC \citep{Gilboa1989}, for instance, the problem of determining an equilibrium with specific properties. 
In general, an equilibrium may not even exist when players solve parametrized non-convex problems. For example, when players solve parametrized integer programs, \citet{Carvalho2018Complexity,Carvalho2020computing} proved that deciding if a Nash equilibrium exists is \SigmaTwoP.

In the context of Stackelberg games with a single leader, however, the standard solution concept is the one of Stackelberg equilibrium. The seminal work of \citet{jeroslow1985polynomial} proved that the complexity of determining if a sequential game admits an equilibrium rises one level up in the polynomial hierarchy for every additional round.

\section{Definitions and Background}
\label{sec:Definitions}

This section provides the basic notations and definitions we employ throughout the paper. As a standard notation in game theory, let the operator $\left( \cdot \right)^{-i}$ denote $\left( \cdot \right)$ except $i$. For any pair of vectors  $\nu \in\R^t$ and $\tau\in\R^t$, let $\nu \perp \tau$ be equivalent to $\nu^\top \tau = 0$. Given $M\in\R^{t\times t}$ and $q\in\R^t$, the linear complementarity problem (\LCP) is the problem of finding, if any, a vector $\nu \in \R^t$ such that $0 \leq \nu \perp (M\nu+q) \geq 0$ \citep{Facchinei2015b,Facchinei2015a,Cottle2009}. 
We say that the optimization problem in $y\in\R^{n_f}$ has a \emph{simple parameterization with respect to $x\in\R^{n_\ell}$} if the problem is in the form of $\min_{y\in\R^{n_f}} \{f(y) + (Cx)^\top y : y \in \mathscr{F}, Ax + By \leq b \}$, where $C$, $A$, $B$, $b$ are matrices and vectors of appropriate dimensions,  $\mathscr{F}\subseteq \R^{n_f}$, and $f:\R^{n_f}\to\R$. 

\subsection{Simultaneous Games}

\begin{definition}[Simultaneous “Nash” Game]\label{def:NashGame}
	A \emph{simultaneous game} $P$ among $n$ players is a finite tuple of optimization problems $P = \left( P^1,\dots,P^n \right)$, where each player $i$ solves $P^i=P^i(x^i,x^{-i})=\min_{x^i\in\R^{n_i}} \lbrace f^i(x^i,x^{-i}) : x^i \in \mathscr{F}^i \}$ with $f^i$ and $\mathscr{F}^i$ being the objective function and the feasible set of $i$, respectively. The game has \emph{complete information} if every player knows    $\mathscr{F}^1,\dots,\mathscr{F}^n$ and $f^1,\dots,f^n$.
\end{definition}

Depending on the structure of each optimization problem $P^i$, we characterize the game as
\begin{enumerate*}[label=(\roman*)]
	\item  \emph{simple} if, for every player $i$, $f^i(x^i,x^{-i}) =  \frac{1}{2} (x^i)^\top Q^ix^i + (c^i)^\top x^i + (C^i x^{-i})^\top x^i$ where $Q^i$ is a positive semi-definite matrix, and $c^i$ and $C^i$ are a vector and a matrix of appropriate dimensions, respectively, or
	\item \emph{linear}, if the game is simple and $Q^i= 0$ for all $i$, or
	\item \emph{facile}, if the game is simple and $\mathscr{F}^i$ is a polyhedron for any $i$.
\end{enumerate*} 
Furthermore, for any player $i$, we call any $x^i \in \mathscr{F}^i$ a \emph{pure} strategy. If a player randomizes over its pure strategies, we call the resulting strategy \emph{mixed}. Formally, a mixed strategy $\sigma^i$ for player $i$ is a probability distribution over $\mathscr{F}^i$.

\begin{definition}[Nash equilibrium]\label{Def:MNE}
	Given a simultaneous game $P$, $\sigma=(\sigma^1,\dots,\sigma^n)$ is an \MNE for $P$ if $\mathbb{E}_{X\sim\sigma}\left[f^i(X^i, X^{-i})\right]  \leq \mathbb{E}_{X^{-i}\sim\sigma^{-i}}\left[ f^i(\tilde x^i, X^{-i})\right]$ for any player $i$ and for any $\tilde x^i \in \mathscr{F}^i$. If $\sigma^i$ has a {\em singleton} support for any $i$, then $\sigma$ is a \PNE.
\end{definition}

Intuitively, \cref{Def:MNE} states that in a Nash equilibrium $\sigma=(\sigma^1,\dots,\sigma^n)$, no player $i$ can \emph{unilaterally deviate} from $\sigma^i$ to any other strategy without increasing the expectation of its objective function value, i.e., without being strictly worse off compared to the equilibrium strategy. 

Whenever $P$ is a facile Nash game, each player's problem is convex in its variables. Therefore, the set of points $x=(x^1,\dots,x^n)$ concurrently satisfying the $KKT$ conditions of \emph{every} problem $P^i$ form the set of \PNEs for $P$. Equivalently, these conditions could be reformulated as an \LCP involving a matrix $M$ and a vector $q$ such that every solution to $0 \leq x \perp (Mx+q) \geq 0$ is a \PNE for $P$ and every \PNE of $P$ solves the \LCP \citep{Cottle2009}.

\subsection{Stackelberg Games}
In contrast to simultaneous games, the players of Stackelberg games play in two rounds \citep{candler1982linear}. Let $w \in\R^{n_\ell}$ be the leader's variables. After the leader plays, its followers compete in a simultaneous game parametrized in the leader's variables $w$. Specifically, in this paper, the followers play a simultaneous game $P(w)$ that has a simple parameterization with respect to $w$, i.e., each problem $P^1(w), \dots, P^n(w)$ has a simple parameterization with respect to $w$.

\begin{definition}[Stackelberg game] \label{Def:bilevel}
	Let $P(w)$ be a simultaneous game with a simple parametrization with respect to $w \in\R^{n_\ell}$, $\sol(P(w))$ be its solution set, $f$ be a function such that $f:\R^{n_\ell+n_f}\to\R$, and $\mathcal{S}\subseteq \R^{n_\ell+n_f}$. A {\em Stackelberg game} $S$ is the optimization problem
	$\displaystyle S=\min_{w\in\R^{n_\ell}, y\in\R^{n_f}} \lbrace f(w,y) : (w,y) \in \mathcal{S},\; y \in \sol(P(w)) \rbrace$. 
\end{definition}

If $P(w)$ is a facile simultaneous game with a simple parameterization with respect to the leader's variables $w$, $\mathcal{S}$ is a polyhedron, and $f(w,y)$ is a linear function, we say $S$ is a \emph{simple} Stackelberg game. From a game-theory perspective, the solution of a Stackelberg game is a so-called \emph{subgame perfect Nash equilibrium}. 
We remark that \cref{Def:bilevel} implies that the Stackelberg game is \emph{optimistic}, i.e., if $P(w)$ has more than one solution, then the followers collectively select the solution $y \in \sol(P(w))$ maximizing the leader's objective function \citep{dempe_solution_2014}. The optimistic assumption may be quite restrictive when $|\sol(P(w))|>1$, since it exposes the issue of equilibria selection for the followers' game. However, it also has a strong economic interpretation. The leader can, in fact, persuade followers to play a favorable solution by transferring utility, i.e., by paying followers an arbitrarily small value. Whenever the followers solve strictly convex quadratic programs, as in our application in \cref{sec:Compute}, the optimistic assumption is not restrictive since there is always a unique followers' equilibrium (i.e., $|\sol(P(w))|=1$).
Recently, \citet{Basu2018a} provided an extended formulation for the feasible region of a simple Stackelberg game. Specifically, the authors proved that the feasible region $\{ (w,y) \in \mathcal{S},\; y \in \sol(P(w))\}$ of a simple Stackelberg game is the union of finitely many polyhedra. We will later employ this result to prove that the convex hull of the union of these polyhedra is the space of mixed strategies for each \NASP player.

\subsection{\NASPs} 
Combining the previous definitions, a \NASP is a simultaneous game with complete information where each player solves a simple parametrized Stackelberg game. 

\begin{definition}[NASP]\label{Def:NASP}
	A \NASP $N=(S^1,\dots,S^n)$ is a complete-information simultaneous game among $n$ players, where each player $i$ solves the simple Stackelberg game $S^i(x^i,x^{-i})=\min_{\substack{x^i\in\R^{n_\ell+n_f}}} \lbrace (c^i)^\top x^i + (C^ix^{-i})^\top x^i : x^i \in \mathscr{F}^i \rbrace$ with $\mathscr{F}^i:=\{x^i=(w^i,y^i): x^i \in \mathcal{S}^i,\; y^i \in \sol(P(w^i))\}$, and $\mathcal{S}^i$ being a polyhedron.
\end{definition}

In \cref{Def:NASP}, we let $x^i \in \R^{n_\ell+n_f}$ be a joint representation of the leader's and the followers' variables for the $i$-th \NASP player. We refer to $\mathscr{F}^i$ as the $i$-th player feasible region, and we say that $\mathscr{F}^i$ is \emph{bounded} if $\mathcal{S}^i$ is a polytope and $\sol(P(w^i)$ is finite.

\begin{remark}
	The Nash equilibrium of \cref{Def:MNE} directly applies to the \NASPs of \cref{Def:NASP}. We remark that, based on our previous definitions, we only consider \NASPs where each Stackelberg game fulfills the optimistic assumptions.
	Therefore, when we claim a \NASP does not have an \MNE, we claim that no \MNE exists when each leader's followers select the most favorable equilibrium for the leader.
	If there is no \MNE satisfying the optimistic assumption, there \emph{might} exist an \MNE where the players do not fulfill such assumption.
\end{remark}

\section{Hardness of Finding a Nash Equilibrium}\label{sec:Hardness}

We characterize the computational complexity associated with deciding whether a \NASP admits an \MNE or not as \SigmaTwoP.  This result holds even for the simplest of \NASPs called a \emph{trivial} \NASP $N=(S^1, S^2)$. In the latter, two leaders have a single follower each, and each follower solves a parametrized linear problem. We formalize our results in \cref{thm:Hardness}, \cref{thm:BndObv}, and \cref{thm:MixHard}.

\begin{theorem}\label{thm:Hardness}
	It is \SigmaTwoP to decide if a trivial \NASP has a \PNE.
\end{theorem}
\begin{corollary}\label{thm:BndObv}
	If each player's feasible set in a \NASP is bounded, an \MNE exists.
\end{corollary} 
\begin{theorem}\label{thm:MixHard}
	It is \SigmaTwoP to decide if a trivial \NASP has an \MNE.
\end{theorem}

\begin{proof}[Proof of \cref{thm:Hardness}.]
\citet[Theorem 2.1]{Carvalho2018Complexity} showed that deciding if a \PNE exists is \SigmaTwoP in two-player simultaneous games where (i.) each player solves a binary integer program with bounded variables, and (ii.) $f^i$ is linear in $x^i$ and parametrized in $x^{-i}$. We employ this result by formulating the Stackelberg game associated with each player as a parametrized integer program.
First, we can employ the reformulation proposed in \citet{Basu2018a}; it requires that the parametrized integer program has a bounded feasible region and is thus compatible with (i.).
Second, we can employ the reformulation proposed in \citet{audet1997links}; it requires that all integer variables are binary, and thus it is also compatible with (i.). Since the objective of each player $i$ is linear in $x^i$, the result of \citet{Carvalho2018Complexity} extends to trivial \NASPs.
\end{proof}

Under an assumption of boundedness, we also show that a trivial \NASP always admits an \MNE. 
\begin{proof}[Proof of \cref{thm:BndObv}.]
Let $\mathscr{F}^i$ be bounded for each player $i$. If $f^i$ is linear in $x^i$, there always exists an optimal solution that is an extreme point of $\conv(\mathscr{F}^i)$. Since the feasible set $\mathscr{F}^i$ of each simple Stackelberg game is a union of polyhedra \citep{Basu2018a},  $\conv(\mathscr{F}^i)$ is a polyhedron. Furthermore, since $\mathscr{F}^i$ is bounded, $\conv(\mathscr{F}^i)$ is precisely a polytope, and the strategies of player $i$ are the finitely many extreme points of $\conv(\mathscr{F}^i)$. Since each player has finitely many strategies, the game is finite, and it always admits an \MNE \citep{Nash1950}.
\end{proof}

We complement the previous result with \cref{thm:MixHard}, where we show that if at least one player's feasible region is not bounded, deciding on the existence of an \MNE is $\Sigma^p_2$-hard.
In contrast to the proof of \cref{thm:Hardness}, we could not apply any of the result of \citet{Carvalho2020computing,Carvalho2018Complexity} since the reformulated players' integer programs: (i.) are not bounded, preventing us from applying the reduction of \citet{Basu2018a},  and (ii.) employ general integer variables as opposed to binary variables, preventing us from using the reduction of \citet{audet1997links}. 
Based on the one of \citet{Carvalho2020computing,Carvalho2018Complexity}, we provide a novel proof that reduces the problem of computing an \MNE in a \NASP to the Subset Sum Interval (\SSI) problem.
While the full proof and its associated claims are available in the electronic companion, we provide a proof sketch below.

\begin{definition}[\SSI]\label{Def:SubSum}
	Given $q_1,\ldots,q_k, p, t, k \in \Z_+$ and $\log_2 (t-p)\leq k$, does there exist an integer $s \in \Z$ with $p \leq s < t $, such that, for all $I \subseteq \{ 1,2,\dots,k \}$, then $\sum_{i\in I}q_i \neq s$?
\end{definition}
The term $t-p$ can be a power of $2$, and the \SSI would then asks if there exists an $r\in\Z_+$ such that $2^r = t- p$. \citet{eggermont2013} proved that, given $r \in \Z_+$ such that $t-p = 2^r$, the problem is $\Sigma^p_2$-complete.

\begin{proof}[Proof Sketch of \cref{thm:MixHard}.]

We reduce \SSI to the problem of deciding on the existence of an \MNE for a trivial \NASP. Let $Q = \sum_{i=1}^k q_i$. 
We explicitly formulate the Stackelberg games associated with each player of a trivial \NASP and denote them as the Latin and the Greek game, respectively. Let $w$ be the Latin leader's decision variables, and let $y$ be its follower variables. Similarly, let $\xi$ be the Greek's leader variables and $\chi$ its follower ones. The Latin player solves the problem 
\begin{subequations}
	\begin{align}
		\max_{\substack{
		w_0,\dots,w_{k+3r+1}
		\in \R\\
		y_0,\dots,y_k
		\in \R
		}} \quad\quad                      &                          
		\frac{w_0}{2} + \sum_{i=1}^k q_iw_i  + 2(Q+1)\xi_{r+1}w_{k+3r+1}
		\nonumber \\
		                                    & \qquad                   
		- (Q+1)\left( 
		\sum_{i=1}^r 2^{i-1} w_{k+i} + pw_{k+3r+1}
		\right)\\
		\text{s.t.}  \quad             
		w_i 
		\quad                               & \geq \quad               
		0, \quad y_i  \quad \geq\quad 0, \quad w_i \quad \leq\quad 1
		\label{eq:MLat:xpos}
		                                    & \forall\, i = 0,\dots,k, \\
		w_{k+3r+1}                          
		\quad                               & =\quad                   
		w_{k+2r+i}  \label{eq:MLat:xeq}
		                                    & \forall\, i=1,\dots,r,   \\
		w_{k+3r+1}                          
		\quad                               & =\quad                   
		p + \sum_{i=1}^r 2^{i-1} w_{k+r+i}, & \label{eq:MLat:xbin}     
		\\
		\frac{w_0}{2} + \sum_{i=1}^k q_iw_i 
		\quad                               & \leq\quad                
		w_{k+3r+1},                         & \label{eq:MLat:knap}     
		\\
		(w_{k+i}, w_{k+r+i}, w_{k+2r+i})    
		\quad                               & \in\quad                 
		\Big \{ (h,y,w)\in\R^3_+:  \left\{h=w,  y =1 \right\} \cup \left\{ h=y=0 \right\} 
		\Big \}
		\label{eq:MLat:S}
		                                    & \forall\, i=1,\dots,r,   \\
		\left( y_0,\dots,y_k \right)        
		\quad                               & \in\quad                 
		\arg\min_{y} \left\{ 
		\sum_{i=0}^k y_i: \begin{array}{l}
		y_i \geq -w_i \\ y_i \geq w_i-1
		\end{array}\forall\, i=0,\dots, k
		\right\}. \label{eq:MLat:bil}
	\end{align}
	The Greek player solves the problem
	\begin{align}
		\max_{\substack{
		\xi_0,\dots,\xi_{r+1}%
		\in \R\\
		\chi_1,\dots,\chi_r%
		\in \R
		}} \quad\quad & (1-w_0)\xi_0                                                                    \\
		\text{s.t.}  \quad
		\xi_i \quad    & \geq\quad 0, \quad     \chi_i \quad \geq\quad 0, \quad \xi_i \quad  \leq\quad 1 
		               & \qquad \forall \, i=1,\dots,r,\label{eq:Gr:chipos}                                     
		\\
		p + \sum_{i=1}^r 2^{i-1}\xi_i              
		\quad          & =\quad                                                                          
		\xi_{r+1},\label{eq:Gr:xibin} \\
		\left( \chi_1,\dots,\chi_r \right)         
		\quad          & \in\quad                                                                        
		\arg\min_{\chi} \left\{ 
		\sum_{i=1}^r \chi_i: \begin{array}{l}
		\chi_i \geq -\xi_i \\ \chi_i \geq \xi_i-1
		\end{array}\forall\, i=0,\dots, r
		\right\}.\label{eq:Gr:bil}
	\end{align}
	\label{eq:MixHard}
\end{subequations}

We show that the \NASP in \cref{eq:MixHard} has an \MNE if and only if the corresponding \SSI problem has an answer YES. 
First, we prove that the formulation ensures that $w_0,\ldots,w_k$ and $\xi_0,\ldots,\xi_r$ are binary (\cref{cl:MNE:Disj}). 
Considering the Greek player, if $w_0=1$, no matter what the Greek player does, its objective function value necessarily equals $0$, the smallest value it can take. 
We construct an \MNE  in this case. 
If $w_0 = 0$, then $\xi_0$ can be arbitrarily large, and there always exists a profitable unilateral deviation for the Greek player. 
Therefore, the problem of checking the existence of an \MNE collapses to the problem of determining whether the Greek player can induce the Latin player to choose $w_0=0$ or not. 
By observing the last two terms in the Latin leader's objective function, we prove that the Latin player plays optimally by always choosing $w_{k+3r+1} = \xi_{r+1}$. Indeed, if $w_{k+3r+1} \neq \xi_{r+1}$, the Latin player acts suboptimally (\cref{cl:MNE:LatCons}).

On the one hand, if the \SSI instance is a \emph{YES} instance (i.e., it admits a solution), then let $\xi_{r+1} = s$.
Since $s$ cannot be represented as a subset-sum, the Latin player can never ensure that \cref{eq:MLat:knap} is satisfied with equality. Even if $w_0=1$, it only adds $0.5$ to the LHS of \cref{eq:MLat:knap}. However, the Latin player would be better off by selecting $w_0 =1 $, as $w_0$ appears in the Latin leader's objective and can improve the objective function's value by $0.5$.

On the other hand, if the \SSI instance is a $NO$ instance, no matter what value the Greek player assigns to $\xi_{r+1}=w_{k+3r+1}$, the Latin player will always ensure that \cref{eq:MLat:knap} is satisfied with equality by choosing $w_0 = 0$. 
Choosing $w_0=1$ is not optimal for the Latin player since it would improve the objective only by $0.5$ while enforcing $w_0 = 0$ would improve the objective at least by $1$. Therefore, enforcing $w_0=0$ is optimal for the Latin player, and the Greek player will never have a profitable deviation from any feasible strategy.
\end{proof}

\section{The Enumeration Algorithm} 
\label{sec:Algo}

Although we proved that deciding on the existence of Nash equilibria is \SigmaTwoP, in this section, we present an effective and exact algorithm to compute \MNEs.
Our algorithmic scheme exploits the structure of the (non-convex) feasible region $\mathscr{F}^i= \{x^i=(w^i,y^i): x^i \in \mathcal{S}^i,\; y^i \in \sol(P(w^i))\}$ of each player $i$. While $\mathcal{S}^i$ is a polyhedron, $\sol(P(w^i))$ is the parametrized set of solutions for the followers' game. We also recall that $\sol(P(w^i))$ is the set of points satisfying all the $KKT$ conditions associated with the followers' optimization problems. Assuming such $KKT$ conditions are expressed as complementarity conditions, then
\begin{align}
	\mathscr{F}^i= \{x^i=(w^i,y^i) : x^i \in \mathcal{S}^i, \; z^i=M^ix^i+q^i, \; 0\leq x^i_j \perp z^i_j \ge 0 \;\;\;  \forall\,j\in\mathcal{C}^i   \} . 
	\label{eq:BilevelSet}   
\end{align}
In the above reformulation, we let $\mathcal{C}^i$ be the set of indices for the complementarity conditions $0\leq x^i_j \perp z^i_j \ge 0$ describing $\sol(P(w^i))$, and $q^i$ and $M^i$ be a vector and a matrix of appropriate dimensions, respectively. The set $\mathscr{F}^i$ also has an equivalent representation as the union of finitely many polyhedra.
Let $\mathcal{O}^i=\{x^i : x^i \in \mathcal{S}^i, \; z^i=M^ix^i+q^i, \; x^i_j \ge 0, \; z^i_j \ge 0 \;\;\;  \forall\,j\in\mathcal{C}^i   \}$ be the \emph{polyhedral relaxation} of $\mathscr{F}^i$. For every player $i$, the finitely many polyhedra whose union is $\mathscr{F}^i$ are given by the \emph{complementarity polyhedra} of \cref{Def:CompPoly}.

\begin{definition}[Complementarity Polyhedron]
\label{Def:CompPoly}
	Given a player $i$ and a binary vector $\theta^i \in \{0,1\}^{|\mathcal{C}^i|}$, the {\em complementarity polyhedron} corresponding to $\theta^i$ is 
	$$\mathscr{P}^i(\theta^i)=
	\{ (x^i,z^i) \in \mathcal{O}^i : x^i_{j} \leq 0 \;\; \forall \; j : \theta^i_j = 0 \} \cap
	\{ (x^i,z^i) \in \mathcal{O}^i : z^i_j \leq 0 \;\; \forall \; j :\theta^i_j = 1 \}.
	$$
\end{definition}

In other words, for any player $i$, $\theta^i$ identifies whether the $j$-th complementarity conditions is active ($x^i_{j} \leq 0$) or not ($z^i_j \leq 0$) in the polyhedron $\mathscr{P}^i(\theta^i)$. Therefore, $\mathscr{F}^i$ is the finite union of \emph{all} the complementarity polyhedra, i.e., 
$\mathscr{F}^i = \bigcup_{ \theta^i \in \{0,1\}^{|\mathcal{C}^i|}} \mathscr{P}^i(\theta^i)
$. Let $k^i$ be the number of the non-empty complementarity polyhedra associated with player $i$.
The core idea behind our algorithm is to reformulate the \NASP by letting each player select its strategies from $\cl \conv(\mathscr{F}^i)$, i.e., the closure of the convex hull of $\mathscr{F}^i$, instead of $\mathscr{F}^i$.  In other words, we \emph{convexify} the \NASP by letting players select their strategies from $\cl \conv(\mathscr{F}^i)$. We prove that any \PNE of the convexified \NASP maps to an equivalent \MNE in the original non-convex \NASP.
In practice, this idea requires, for each player $i$, the explicit enumeration of the $k^i$ complementarity polyhedra and the computation of the closure of their convex hull.
We will employ the result of \citet{Balasa} to compute the latter.

\begin{theorem}[\citet{Balasa}]\label{thm:Balas}
	Given $k$ polyhedra $K_j = \lbrace x\in\R^n : A_jx \leq b_j \rbrace$ for $j=1,\dots k$, then $\cl \conv( \bigcup_{j=1}^k K_j)=
	\{ x \in \R^n: \exists (x_1,\ldots, x_k,\delta) \in (\R^n)^k \times \R^k : 
	A_jx_j \leq \delta_j b_j, \;
	\sum_{w=1}^kx_w =x, \;
	\sum_{w=1}^k \delta_w =1, \;
	\delta_j \geq 0,\;
	\forall j =1,\dots,k
	\} 
	$.
\end{theorem}

The last ingredient of our algorithm is the result of \citet[Theorem 2.7]{stein_separable_2008} for \emph{separable} simultaneous games, i.e., simultaneous games where $f^i(x^i,x^{-i})$ is a separable polynomial for every player $i$. In a separable simultaneous game, \citet{stein_separable_2008} proved that every mixed strategy $\sigma^i$ of player $i$ has either finite support or a finite support equivalent. Since \NASPs are indeed separable simultaneous games, we restrict our attention to \MNEs where the players randomize over a finite number (precisely, at most $k^i$)  of pure strategies.

\paragraph{The Enumeration Algorithm. }
We formalize our enumeration scheme in \cref{Alg:FullEnumeration}. 
Given a \NASP instance $N=(S^1,\dots,S^n)$, the algorithm returns either an \MNE for $N$ or a proof of its non-existence, i.e., a \texttt{no} certificate. \cref{Alg:FullEnumeration} exploits the equivalent convex representation of $N$, where each player's feasible region $\mathscr{F}^i$ is replaced by $\tilde{\mathscr F}^i:=\cl \conv(\mathscr F^i)$. In \cref{Alg:FullEnumeration:hull}, the algorithm retrieves the sets $\tilde{\mathscr F}^i$ by enumerating the $k^i$ complementarity polyhedra and computing the closure of their convex hull through the extended formulation of \cref{thm:Balas}.
In \cref{Alg:FullEnumeration:fac}, the algorithm reformulates $N$ into the equivalent simultaneous game $\tilde N$, where each player $i$ solves $\tilde{S}^i := \min_{x^i\in\R^{n_\ell+n_f}} \lbrace f^i(x^i,x^{-i})=(c^i)^\top x^i + (C^ix^{-i})^\top x^i : x^i \in \tilde{\mathscr{F}}^i \rbrace$. At this step, $\tilde N$ is a facile simultaneous game, and we can formulate an \LCP to determine its \PNEs. 
\begin{algorithm}[!ht]
    \DontPrintSemicolon
	\caption{Enumeration Algorithm for \NASPs \label{Alg:FullEnumeration}}
	\KwData{A \NASP instance $N = (S^1,\ldots,S^n)$.}
	\KwResult{ Either: (i.)
		$(\hat x^i,p^i)$ for every $i$, or (ii.) \texttt{no}.}
	\For{$i\gets1$ \KwTo $n$}{	
		$\tilde{\mathscr F}^i \gets \cl \conv (\mathscr{F}^i) $ by applying \cref{thm:Balas}\label{Alg:FullEnumeration:hull}\;
		$\tilde{S}^i \gets \min_{x^i\in\R^{n_\ell+n_f}} \lbrace f^i(x^i,x^{-i})=(c^i)^\top x^i + (C^ix^{-i})^\top x^i : x^i \in \tilde{\mathscr{F}}^i \rbrace$\;
	}
	\If{$\exists$ a \PNE $(\tilde x^1,\ldots,\tilde x^n)$ of $\tilde N = (\tilde S^1,...,\tilde S^n)$ \label{Alg:FullEnumeration:fac}} 
	{
		\For{$i\gets1$ \KwTo $n$}{	
			\lIf{$\tilde x^i \in \mathscr{F}^i$} {$\exists j :\hat x^i_j \gets \tilde x^i$, $p^i_j \gets 1$ \label{Alg:FullEnumeration:pure}}
			\lElse{
				$\tilde x^i = \sum_{j=1}^{k^i} p^i_j \hat x^i_j$ for $\hat x^i_1,\dots, \hat x^i_{k^i} \in \mathscr{F}^i$ with $p^i_j\geq 0$ and $\sum_{j=1}^{k^i}p^i_j=1$\label{Alg:FullEnumeration:mix}
			}
		}
		\KwRet{($\hat x^i_j$, $p^i_j$) for each player $i=1,\dots,n$ and $j = 1,...,k^i$}\;
	} 
	\lElse{\KwRet{\texttt{no}} \label{Alg:FullEnumeration:no}}
\end{algorithm}
We claim that any \PNE of $\tilde N$ maps to an equivalent \MNE in $N$.
Let $\tilde x=(\tilde{x}^1,\dots,\tilde{x}^n)$ be a \PNE of $\tilde N$. If $\tilde{x}^i \in \mathscr{F}^i$ for each player $i$, then $i$ plays the pure strategy $\tilde{x}^i$ in the \MNE of $N$; therefore, $\tilde{x}$ is necessarily a \PNE for $N$ (\cref{Alg:FullEnumeration:pure}). Otherwise, if there exists a player $i$ such that $\tilde{x}^i \notin \mathscr{F}^i$, then $\tilde{x}^i \in \cl \conv (\mathscr{F}^i) \backslash \mathscr{F}^i$ and we claim that $\tilde{x}^i$ is equivalent to a mixed strategy (\cref{Alg:FullEnumeration:mix}). 
Indeed, $\tilde{x}^i$ is either a convex combination of points in $\mathscr{F}^i$ or a limit of such points. Specifically, the points in the support of the convex combination belong to $\mathscr{F}^i$, i.e., they are pure strategies; the convex combination's coefficients $p^i$ are the probabilities associated with each point in the support. 
In practice, the probabilities $p^i$ are precisely the values of the $\delta$ variables provided by \cref{thm:Balas}.
We provide a visualization of this convexification method in \cref{fig:Enumeration}. Finally, we prove the properties of \cref{Alg:FullEnumeration} in \cref{thm:Alg}.

\begin{figure}[ht]
	\captionsetup{justification=centering}
			    
	\centering
			        
	\begin{subfigure}[b]{0.45\textwidth}
		\centering
		\includegraphics[width=\textwidth]{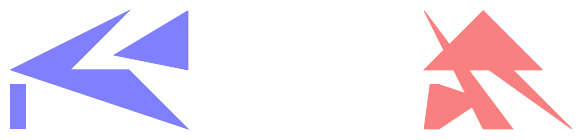}
		\caption{\cref{Alg:FullEnumeration:hull} of \cref{Alg:FullEnumeration}: each player's $i$ feasible region $\mathscr{F}^i$ is a non-convex union of finitely many polyhedra. }
		\label{fig:feas}
	\end{subfigure}
	\hfill
	\begin{subfigure}[b]{0.45\textwidth}
		\centering
		\includegraphics[width=\textwidth]{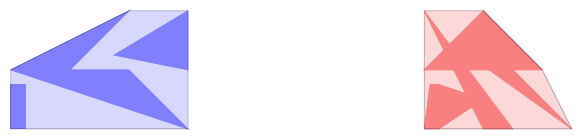}
		\caption{\cref{Alg:FullEnumeration:hull} of \cref{Alg:FullEnumeration}: for each $i$, the algorithm computes $\cl \conv(\mathscr{F}^i)$ via \cref{thm:Balas}. }
		\label{fig:hull}
	\end{subfigure}
	\bigskip 
			        
	\begin{subfigure}[b]{0.45\textwidth}
		\centering
		\includegraphics[width=\textwidth]{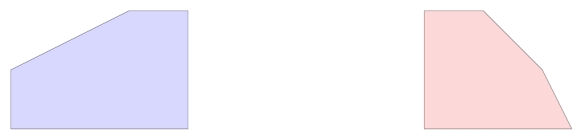}
		\caption{\cref{Alg:FullEnumeration:fac} of \cref{Alg:FullEnumeration}: each player solves a parametrized convex problem over $\cl \conv(\mathscr{F}^i)$.    %
		}
		\label{fig:game}
	\end{subfigure}
	\hfill 
	\begin{subfigure}[b]{0.45\textwidth}
		\centering
		\includegraphics[width=\textwidth]{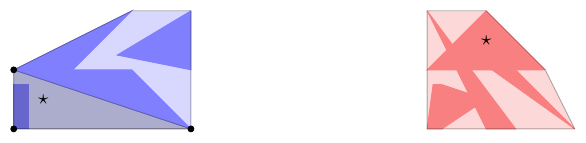}
		\caption{\cref{Alg:FullEnumeration:pure,Alg:FullEnumeration:mix} of \cref{Alg:FullEnumeration}: the \PNE $\star$ of $\tilde N$ is an \MNE for $N$, since each $\tilde{x}^i$ is a convex combination of strategies in $\mathscr{F}^i$. }
		\label{fig:hullGame}
	\end{subfigure}
	\centering
	\caption{A representation of \cref{Alg:FullEnumeration} for two players (blue and red).} 
	\label{fig:Enumeration} 
\end{figure}

\begin{theorem} \label{thm:Alg}
	\cref{Alg:FullEnumeration} terminates in a finite number of steps and returns either
	\begin{enumerate*}[label=(\roman*)]
		\item an \MNE for $N$ so that each player $i$ plays the strategy $\hat x^i_j$ with probability $p^i_j$, or
		\item \texttt{no} if $N$ has no \MNE.
	\end{enumerate*}
\end{theorem}
\begin{proof}[Proof of \cref{thm:Alg}.]
The algorithm terminates in a finite number of steps since all the \texttt{for} loops have $n$ iterations, $k^i \le 2^{\mathcal{C}^i}$ for any $i$, and \cref{Alg:FullEnumeration:fac} is an \LCP.

\paragraph{Proof of Statement (i.). } 
\begin{subequations}
	If \cref{Alg:FullEnumeration} returns $(\hat x^i,p^i)$, then  \cref{Alg:FullEnumeration:fac} finds a \PNE $\tilde x$ for $\tilde N$. 
	Since each player's objective function $f^i$ is separable and linear in $x^i$,  $\tilde x^i$ is a distribution with finite support over $\mathscr{F}^i$ \citep{stein_separable_2008}.
	Therefore, for every player $i$,
	\begin{align}
		\mathbb{E}_{\hat{X}\sim p}\left[ \left( c^i+C^i\hat X^{-i} \right)^\top \hat X^i \right] 
		\quad=\quad                                                                              
		\sum_{j'}\sum_{j=1}^{k_i}                                                                
		p^{-i}_{j'} p^i_j                                                                        
		\left(                                                                                   
		c^i + C^i\hat x^{-i}_{j'}                                                                
		\right)^\top                                                                             
		\hat x^i_j                                                                               
		\quad=\quad                              
		\left(   c^i + C^i\tilde x^{-i} \right)^\top \tilde x^i.                                 
		\label{eq:Alg1}                                                                          
	\end{align}
	Assume $(\hat x^i,p^i)$ is not an \MNE for $N$, or, equivalently, there exists player $i$ that can profitably and unilaterally deviate from $(\hat x^i,p^i)$  to $(\dagger \hat x^i,\dagger p^i)$. By definition, $\sum _{j=1}^{k^i}(\dagger p^i_j\dagger \hat x^i_j)$ is a pure strategy for $i$ in $\tilde N$. By leveraging the linearity of $f^i$ in $x^i$ for any player $i$, we can show that $\sum _{j=1}^{k^i}(\dagger p^i_j\dagger \hat x^i_j)$ is also a profitable deviation for $\tilde P^i$ in $\tilde N$ since 
	\begin{align}
		\left(   c^i + C^i\tilde x^{-i} \right)^\top \tilde x^i
		  & \quad= 	\quad                                            %
		\sum_{j'}\sum_{j=1}^{k_i} p^{-i}_{j'}p^i_j \left(   c^i + C^i\hat x^{-i}_{j'} \right)^\top \left( \hat x^i_j \right) & \nonumber
		\\ & \quad \geq 	\quad                              
		\sum_{j'}\sum_{j=1}^{k^i}p^{-i}_{j'} \dagger p^i_j \left(   c^i + C^i\tilde x^{-i}_{j'} \right)^\top \left( \dagger \hat x^i_j \right) \nonumber \\
		  & \quad=\quad  \left(   c^i + C^i\left( \sum_{j'}p^{-i}_{j'}\tilde x^{-i}_{j'} \right) \right)^\top \left( \sum_{j=1}^{k^i}\dagger p^i_j  \dagger \hat x^i_j\right) 
		 \nonumber \\     
		&\quad=\quad \left(   c^i + C^i \tilde x^{-i} \right)^\top \left( \sum_{j=1}^{k^i}\dagger p^i_j  \dagger \hat x^i_j\right). \nonumber
		\label{eq:Alg1Plug}
	\end{align}
	This last equation proves that there also exists a profitable deviation from $(\hat x^i,p^i)$ in $\tilde N$, contradicting the fact $(\hat x^i,p^i)$  is a \PNE in $\tilde N$.

	\paragraph{Proof of Statement (ii.). }
	To prove this statement, we prove its contrapositive; namely, we show that if $N$ has an \MNE, then \cref{Alg:FullEnumeration:fac} obtains a \PNE for $\tilde N$.
	Let the \MNE of $N$ be $(\hat x^i,p^i)$, and
	let $\tilde x^i = \sum_{j=1}^{k_i} p_j^i \hat{x}^i_j$. Since $(\tilde x^1,\dots,\tilde x^n)$ is both feasible in $\tilde N$ and an \MNE in $N$, then
	$$\sum_{j'}\sum_{j=1}^{k^i} p^{-i}_{j'} p^i_j ( C^i\hat{x}^{-i}_{j'} + c^i)^\top\hat{x}^i_j \quad \leq \quad \sum_{j'} p^{-i}_{j'}  ( C^i\hat{x}^{-i}_{j'} + c^i )^\top \bar x^i \qquad \forall\,  \bar x^i \in \mathscr F_i .$$
	By leveraging the linearity of $f^i$ in $x^i$, we have that $\left( C^i\tilde x^{-i} + c^i \right)^\top\tilde x^i \leq \left( C^i \tilde x^{-i} + c^i \right)^\top \bar x^i$ for any $\bar x^i \in \mathscr F_i$.
	If the previous inequality holds for any $\bar x^i\in \cl\conv(\mathscr F_i)$ and for any player $i$, then $\tilde x$ is also a \PNE of $\tilde N$. 
	First, the inequality holds for any $\bar x^i \in \conv(\mathscr F_i)$.
	Let $\bar x^i = \sum_{j=1}^{k^i} \lambda_j \bar x^i_j$, where $\bar x^i_j \in \mathscr {F}_i$ and $\lambda_j \geq 0$ and $\sum_{j=1}^{k^i} \lambda_j=1$, i.e., $\lambda$ are the coefficients of a convex combination. Consider the $k^i$ inequalities associated with each strategy $\bar x^i_j$ in the support of $\bar{x}^i$. By multiplying such inequalities by the associated non-negative $\lambda_j$ on both sides, we obtain that
	\begin{align}
		\left( C^i\tilde x^{-i} + c^i \right)^\top\tilde x^i
		\quad & \leq\quad 
		\sum_{j=1}^{k^i} \lambda_j \left( C^i \tilde x^{-i} + c^i \right)^\top \bar x^i_j 
		\quad=\quad
		\left( C^i \tilde x^{-i} + c^i \right)^\top \bar x^i.
	\end{align}
	Therefore, the inequality holds for any $\bar x^i \in \conv(\mathscr F_i)$. To extend this result to the closure, i.e., for any $\bar x^i \in \cl \conv (\mathscr F_i)$, we consider a convergent sequence $\bar x^i_1, \bar x^i_2,\dots$ such that $\bar x^i_j \in \conv (\mathscr F_i)$ and $\lim_{j\to\infty} \bar x^i_j = \bar x^i$. For any $j=1,\dots,k^i$ and any $i$
	\begin{align}
		\left( C^i\tilde x^{-i} + c^i \right)^\top\tilde x^i
		\quad                                                  & \leq\quad 
		\left( C^i \tilde x^{-i} + c^i \right)^\top \bar x^i_j &           
		 \nonumber
		\\
		\implies
		\lim_{j\to\infty}
		\left( C^i\tilde x^{-i} + c^i \right)^\top\tilde x^i
		\quad                                                  & \leq\quad 
		\lim_{j\to\infty}
		\left( C^i \tilde x^{-i} + c^i \right)^\top \bar x^i_j & \nonumber 
		\\ 
		\implies
		\left( C^i\tilde x^{-i} + c^i \right)^\top\tilde x^i
		\quad                                                  & \leq\quad 
		\left( C^i \tilde x^{-i} + c^i \right)^\top \left( \lim_{j\to\infty} \bar x^i_j \right)
		\quad  = \quad \left( C^i \tilde x^{-i} + c^i \right)^\top \bar x^i. \nonumber 
	\end{align}
	Thus, $\left( C^i\tilde x^{-i} + c^i \right)^\top\tilde x^i \leq \left( C^i \tilde x^{-i} + c^i \right)^\top \bar x^i$ holds for any $\bar x^i \in \cl\conv(\mathscr F_i)$, and $\tilde x$ is a \PNE of $\tilde N$. 
\end{subequations}
\end{proof}

\begin{remark}\label{rem:convHullisEnuf}
	In the proof of \cref{thm:Alg}, we only assume that any $f^i$ is linear in $x^i$. Whenever this assumption holds, any \NASP admits an \emph{equivalent convex representation} $\tilde N$ so that any \PNE in $\tilde N$ is an \MNE in $N$. Furthermore, we can solve the \LCP in \cref{Alg:FullEnumeration:fac} of \cref{Alg:FullEnumeration} as a Mixed-Integer Program (\MIP), and select an equilibrium optimizing a given objective function. 
\end{remark}

\section{Algorithmic Enhancements}
\label{sec:enhance}
We propose two enhancements to \cref{Alg:FullEnumeration}. First, in \cref{sec:sub:inner}, we introduce an algorithm that iteratively refines a convex approximation of each $\tilde{\mathscr F}^i$.  Second, in \cref{sec:sub:PNE}, we tailor \cref{Alg:FullEnumeration} to specifically retrieve \PNEs, instead of general \MNEs.

\subsection{Inner Approximation Algorithm}
\label{sec:sub:inner}

\cref{Alg:FullEnumeration} necessarily converges to an \MNE or a proof of its non-existence. 
However, from a practical perspective, $k^i$ may be exponential in ${\mathcal{C}^i}$ for any $i$, i.e., $k^i 
=2^{\mathcal{C}^i}$. Thus, computing $\tilde{\mathscr F}^i$ may be prohibitively difficult in practice, as it requires the enumeration of an exponential number of complementarity polyhedra for each player. %
Motivated by this practical issue, we introduce an \emph{inner approximation} algorithm inspired by \cref{Alg:FullEnumeration}. We aim to possibly avoid the expensive enumeration of the complementarity polyhedra by iteratively refining a convex inner approximation of $\tilde{\mathscr F}^i$ for every player $i$. Instead of computing $\tilde{\mathscr F}^i$, i.e., the closure of the convex hull of the union of \emph{all} the complementarity polyhedra for $i$, this algorithm approximates $\tilde{\mathscr F}^i$ by considering the closure of the convex hull of the union of \emph{some} complementarity polyhedra. Let $\hat{\mathscr F}^i \supseteq \tilde{\mathscr F}^i$ be the inner approximation of $\tilde{\mathscr F}^i$.  Starting from $\hat{\mathscr F}^i= \mathcal{O}^i$, we refine $\hat{\mathscr F}^i$ by computing the closure of the convex hull of $\hat{\mathscr F}^i \cup \mathscr{P}^i(\theta^i)$, i.e., the union of $\hat{\mathscr F}^i$ and an additional complementarity polyhedron. This inner approximation scheme enables to iteratively grow the description of $\hat{\mathscr F}^i$, even when considering the union of $\hat{\mathscr F}^i$ with several complementarity polyhedra.
Let $\mathcal{J}^i \subseteq \{0,1\}^{|\mathcal{C}^i|}$ be a \emph{set of complementarity polyhedra} for $i$ so that, for any $\theta^i \in \mathcal{J}^i$, $\mathscr{P}^i(\theta^i)$ is a complementarity polyhedron.

\begin{definition}[Inner Approximation]\label{def:innerapprox}
	The \emph{inner approximation} for $i$ induced by the set of complementarity polyhedra $\mathcal{J}^i$ is
	$
	\mathcal I^i_{\mathcal{J}^i}= \cl\conv \left ( 
	\bigcup_{\theta^i\in \mathcal{J}^i} (\mathscr{P}^i(\theta^i))
	\right ) 
	$.
\end{definition}

For any choice of $\mathcal{J}^i$ and player $i$, we remark that $\mathcal I^i_{\mathcal{J}^i}$ is a polyhedron and, thus, convex.

\begin{algorithm}
	\caption{Inner Approximation Algorithm for \NASPs \label{Alg:InnerApproximation}}
	\KwData{A \NASP instance $N = (S^1,\ldots,S^n)$,  $\mathcal{J}=(\mathcal{J}^1,...,\mathcal{J}^n)$ with $\mathcal{J}^i\subseteq \{0, 1\}^{|\mathcal \mathcal{C}^i|}$.}
	\KwResult{ Either: (i.)
		$(\hat x^i,p^i)$ for every $i$, or (ii.) \texttt{no}.}

	\DontPrintSemicolon
	\SetKwFunction{FMain}{InnerApproximation}
	\SetKwProg{Fn}{Function}{:}{}
	\Fn{\FMain{$N, \mathcal{J}$}}{
		\For{$i\gets1$ \KwTo $n$}{	
			$\hat{ \mathscr{F}}^i \gets \mathcal I^i_{\mathcal{J}^i}$,
			$\tilde{S}^i \gets \min_{x^i\in\R^{n_\ell + n_f}} \lbrace f^i(x^i,x^{-i})=(c^i)^\top x^i + (C^ix^{-i})^\top x^i : x^i \in \hat{\mathscr{F}}^i \rbrace$ \; \label{Alg:InnerApproximation:approx}
		}
		\If{$\exists$ a \PNE $(\tilde x^1,\ldots,\tilde x^n)$ of $\tilde N = (\tilde S^1,...,\tilde S^n)$ \label{Alg:InnerApproximation:tildeN}}
		{
			\If{\texttt{getDeviation}$(S^i, \tilde x)=\emptyset$ for any player $i$ \label{Alg:InnerApproximation:deviation}} {
				$\tilde x^i = \sum_{j=1}^{k^i} p^i_j \hat x^i_j$ for $\hat x^i_1,\dots, \hat x^i_{k^i} \in \mathscr{F}^i$ with $p^i_j\geq 0$ and $\sum_{j=1}^{k^i}p^i_j=1$\;
				\KwRet{($\hat x^i_j$, $p^i_j$) for each player $i=1,\dots,n$ and $j = 1,...,k^i$\label{Alg:InnerApproximation:yes}\; } 
			} 
			\ElseIf{\texttt{getDeviation}$(S^i, \tilde x)=\hat{x}^i$ for some player $i$}{
				$\tilde \theta^i \gets $ the complementarity polyhedron so that $\hat x^i \in \mathscr{P}^i(\theta^i)$, 
				$\mathcal{J}^i \gets \mathcal{J}^i \cup \{\tilde \theta^i\}$\label{Alg:InnerApproximation:encoding}\;
				\KwRet{\texttt{InnerApproximation}$(N,\mathcal{J})$\; \label{Alg:InnerApproximation:rec1}} 
			}
		}
		\Else{
			\If{ $\exists \; \tilde \theta^i : \tilde \theta^i \notin \mathcal{J}^i$ \label{Alg:InnerApproximation:extension}}{
				$\mathcal{J}^i \gets \mathcal{J}^i \cup \{\tilde \theta^i\}$ and \textbf{return} \texttt{InnerApproximation}$(N,\mathcal{J})$\label{Alg:InnerApproximation:failStep}\;
			}
			\lElse{
				\KwRet{\texttt{no} } 
			}
		}
	}
\end{algorithm}

\paragraph{The Algorithm. }In \cref{Alg:InnerApproximation}, we present the inner approximation algorithm $\texttt{InnerApproximation}(N,\mathcal{J})$ to compute an \MNE for a \NASP $N$. 
Let $\mathcal{J}=(\mathcal{J}^1,...,\mathcal{J}^n)$ be a vector containing an arbitrary initial set of complementarity polyhedra for each player $i$. In \cref{Alg:InnerApproximation:approx}, the algorithm computes, for each player $i$, the associated inner approximations $\hat{\mathscr F}^i$. During the first iteration, if we initialize $\mathcal{J}^i=\emptyset$, then $\hat{\mathscr F}^i= \mathcal{O}^i$.
Similarly to \cref{Alg:FullEnumeration}, \cref{Alg:InnerApproximation} formulates a convexified game $\tilde{N}$ (\cref{Alg:InnerApproximation:tildeN}) where each player's feasible region is the polyhedron $\hat{\mathscr F}^i=\mathcal I^i_{\mathcal{J}^i}$. Since $\tilde N$ is a facile simultaneous game, \cref{Alg:InnerApproximation:tildeN} determines a \PNE by solving an equivalent \LCP.

On the one hand, assume that $\tilde N$ admits the \PNE $\tilde x$. In order to check whether $\tilde x$ is an \MNE for $N$, we employ the routine \texttt{getDeviation} (\cref{Alg:InnerApproximation:deviation}). Given $\tilde x$ and a player $i$, the routine computes the optimal solution $\hat{x}^i$ of the Stackelberg game $S^i$ by fixing the other players' choices $x^{-i}$ to $\tilde{x}^{-i}$, i.e., it computes the so-called \emph{best response} $\hat{x}^i = \arg\min_{x^i} S^i(x^i,\tilde{x}^{-i})$. If $f^i(\tilde{x}^i,\tilde{x}^{-i}) = f^i(\hat{x}^i,\tilde{x}^{-i})$, the routine determines that $\tilde{x}^i$ is optimal and returns $\emptyset$ since there exists no profitable deviation for player $i$. Otherwise, if $f^i(\tilde{x}^i,\tilde{x}^{-i}) \ge f^i(\hat{x}^i,\tilde{x}^{-i})$, $\hat{x}^i$ is a profitable deviation for player $i$, and the routine returns $\hat{x}^i$.
If, for any player $i$, there exists no profitable deviation from $\tilde{x}$ in $N$, then $\tilde x$ is also an \MNE for $N$, and the algorithm returns the \MNE in \cref{Alg:InnerApproximation:yes}. Otherwise, at least one player $i$ has a profitable deviation $\hat{x}^i$, and \cref{Alg:InnerApproximation:encoding} refines $\mathcal{J}^i$ by including the polyhedron containing $\hat{x}^i$. Specifically, in \cref{Alg:InnerApproximation:encoding}, the algorithm adds to $\mathcal{J}^i$ the complementarity polyhedron $\tilde \theta^i$ so that $\hat{x}^i \in \mathscr{P}^i(\tilde \theta^i)$. Consequently, the algorithm starts a recursion with the updated $\mathcal{J}$ in \cref{Alg:InnerApproximation:rec1}.

On the other hand, $\tilde N$ may not admit a \PNE as in \cref{Alg:InnerApproximation:failStep}. In this case, we have no information to guide the refinement of the players' inner approximations. Therefore, the algorithm arbitrarily refines at least one $\mathcal{J}^i$ by including at least one complementarity polyhedron  $\tilde{\theta}^i \notin \mathcal{J}^i$ and start a new recursion (\cref{Alg:InnerApproximation:extension}); we will discuss some strategies regarding the selection of $\tilde{\theta}^i$ in \cref{sec:sub:speed}. If no \MNE to $\tilde N$ exists when the players' approximations are exact, i.e., when $\mathcal{J}^i=\tilde{\mathscr{F}^i}$ for any player $i$, then no \MNE in $N$ exists. The results on the correctness and finite termination of \cref{thm:Alg} also extend to \cref{Alg:InnerApproximation}.
 
\paragraph{Hierarchy of Approximations. }In optimization, given a problem and one of its inner approximations, a solution to the inner approximation is also a solution to the original problem. However, this is not the case for Nash equilibria: in \NASPs, the inner approximated game $\tilde N$ may admit an \MNE that is not an \MNE for the original game $N$, as we show in \cref{ex:InnAppCounter}.
\begin{remark}
	\label{ex:InnAppCounter} 
	Consider a \NASP with $n=2$, where the first \emph{Latin} player solves $\min_x  \{ \xi x : x\in\R, x\geq 0\}$, and the second \emph{Greek} player solves $\min_{\xi,\chi} \{ x\xi  :  \xi \in [-5,5], \chi \ge 0,\; \chi\in\arg\min_{\chi} \{ \chi: \chi \ge \xi-1, \chi \ge -\xi-1\}$.
	This \NASP has no \MNE since the \emph{Greek} player is optimal with $\xi = -5$ for any \emph{Latin}'s player decision, and when $\xi = -5$, the \emph{Latin} player's problem is unbounded.
	By explicitly writing the $KKT$ conditions of the Greek follower's problem, the \emph{Greek}'s problem becomes
	\begin{align*}
		\min_{\xi,\chi,\mu} \left \{
		x\xi: \xi \in [-5,5], \; \mu_1+\mu_2 = 1, \; \chi \ge 0, \;
		\begin{array}{cccc}
		0\le & \mu_1 & \perp \chi - \xi + 1 & \ge 0, \\
		0\le & \mu_2 & \perp \chi + \xi + 1 & \ge 0 
		\end{array}
		\right\}.
	\end{align*}
	If $\mathcal{J}^2=\{(0,0),(1,1),(1,0),(0,1)\}$, then $\mathcal I^2_{\mathcal{J}^2}=\tilde{\mathscr F}^2$, and the approximation is exact.
	While the first two polyhedra $(0,0)$ and $(1,1)$ are empty, the remaining two polyhedra can be projected onto the $\xi $ space as $[-5,-1]\cup [1,5]$. 
	Consider the inner approximation $\mathcal{J}^2 = \{(0, 1)\}$ corresponding to the inner approximated game $\tilde N$, where the Latin player solves $\min_x \left \{ \xi x : x\in\R, x\geq 0 \right\}$ and the Greek player solves $\min_{\xi} \left \{ x\xi  : \xi\in\R, \xi \in [1, 5]\right\}$.
	The inner approximation is exact for the \emph{Latin} player, and $\tilde N$ has a \PNE $(\xi,x)=(0,1)$.
	Conversely, consider a game $N$ where the \emph{Greek}'s objective becomes $\min_{\xi} -x\xi$, and the corresponding inner approximation becomes $\xi \in [-5,-1]$. In this case, while $N$ has an \MNE $(\xi,x) = (0, 5)$, $\tilde N$ would have no \MNE.
\end{remark}

\subsection{Computing \PNEs}
\label{sec:sub:PNE}
Several applications demand deterministic \PNEs instead of randomized \MNEs. This practical need motivates an enhancement of \cref{Alg:FullEnumeration} to compute only \PNEs.
As previously mentioned, any strategy $x^i \in \mathscr{F}^i$ is a pure strategy for player $i$. This also implies that the pure strategies are strictly contained in a single complementarity polyhedron. 
We can practically require that each $x^i$ belongs to $\mathscr{F}^i$ by considering the $\delta$ variables associated with the extended formulation of \cref{thm:Balas} and enforcing them to be binary. Since the $\delta$ are the convex multipliers associated with each complementarity polyhedron, whenever there exists a $j \in \{1,\dots,k\}$ so that $\delta_j=1$, the projection of $x$ onto the original space belongs to the $j$-th complementarity polyhedron.
Thus, we can modify \cref{Alg:FullEnumeration} to compute \PNEs by (i.) solving the \LCP associated with $\tilde N$ as a \MIP, and (ii.) enforcing the $\delta$ ($\delta^i$ for any $i$) variables to be binary. Any \PNE for $\tilde{N}$ is necessarily a \PNE for $N$, and if $\tilde{N}$ has no \PNE, $N$ has no \PNE.

\section{Experiments on Energy Markets}
\label{sec:Compute}
This section introduces an energy market model based on \NASPs.
We analyze and compare the performance of our algorithms, and we provide an extensive set of computational results leading to clear and enlightening managerial insights for market regulators.

\subsection{The Energy Model }
Let the color {\color{rosso}red} identify the \NASP's parameters, the color {\color{blue}blue} identify each player $i$ variables, and the color {\color{green}green} identify the variables of the opponents of $i$\footnote{
We employ color coding on symbols for readability. The color-ready version of the paper is available in the online version only.}.
We model an energy market where a set $\mathcal{R}$ of $n$ energy regulators (e.g., governmental agencies) from different geographical regions oversee the operations of their domestic markets and trade energy among themselves. 
Each regulator $r \in \mathcal{R}$ solves the Stackelberg game 
\begin{subequations}
	\begin{align}
		\label{eq:EC:leader}
		\min_{\substack{\qi[r],\ti[r],\qAimpI[r'\to r], \qexp[r]}} \quad\quad &                                                                                                                                                                  
		\sum_{p\in \mathcal{P}^r} \Cemm[r,p]\qi[r,p] -\taxb \ti[r]\qi[r] + \sum_{r'\in\mathcal{R}\setminus r}\piI[r]\qAimpI[r'\to r] - \piI[r]\qexp[r] \\
		\text{s.t.} \qquad \quad                                        & \ti[r] \quad \le\quad  \tiCap[r] \label{eq:EC:cap},                                                                                                              \\
		                                                                      & \DemInt[r] - \DemSlope[r] \left( \sum_{p\in\mathcal{P}^r} \qi[r,p] + \qimp[r] - \qexp[r], \right) \quad \le\quad  \piCeil[r], \label{eq:EC:price}                \\
		                                                                      & \sum_{r'\in\mathcal{R}}\qAimpI[r'\to r] \quad =\quad  \qimp[r],                                                                            \label{eq:EC:imports} \\
		                                                                      & \qi[r] \quad \in\quad  \sol (\text{Simultaneous game among }\mathcal{P}^r)\label{eq:EC:subgame}.                                                                 
	\end{align}
\end{subequations}
The regulator $r \in \mathcal{R}$ matches the domestic demand of energy given by a demand curve with the intercept $\DemInt[r]$ and the elasticity parameter $\DemSlope[r]$. 
Inside each regulator's market, a set $\mathcal{P}^r$ of energy producers act as followers playing a simultaneous Cournot competition on the amount of energy produced (Constraint \cref{eq:EC:subgame}). In other words, given the domestic demand curve associated with each regulator's $r$ geographical area, each producer $p \in \mathcal{P}^r$ decides the quantity of energy units $\qi[r,p]$ to inject into the market depending on \begin{enumerate*}[label=(\roman*)] \item the regulator's policies, \item the parameters of its cost structure, and \item the current price of energy.\end{enumerate*} The regulator $r$ minimizes a cost $\Cemm[r,p]$ on each unit of energy $\qi[r,p]$ produced by $p$ to reduce emissions while concurrently being constrained to keep the domestic energy price under a given threshold $\piCeil[r]$ (Constraint \cref{eq:EC:price}).
Specifically, $\Cemm[r,p]$ is the product of 
\begin{enumerate*}[label=(\roman*)]
	\item the \emph{social cost of carbon}, i.e., the cost incurred due to the emission of one unit of greenhouse gases, and
	\item the \emph{emission factor}, i.e., the amount of greenhouse gases emitted for each unit of energy produced by $p$.
\end{enumerate*}
In order to meet the demand, $r$ can import $\qimp[r]$ units of energy from other markets so that $\qimp[r]$ is the sum of the amount of energy $\qAimpI[r'\to r]$ imported from any other $r' \in \mathcal{R}\setminus r$ (Constraint \cref{eq:EC:imports}). We let $\taxb \in \{0,1\}$ be a tax parameter. If $\taxb=1$, the regulator collects a tax of $\ti[r]$ on each unit of energy produced in the market with a tax cap of $\tiCap[r]$ (Constraint \cref{eq:EC:cap}); since the objective is no longer linear with $\taxb = 1$, we replace the nonlinear product terms with their McCormick envelopes \citep{mccormick_computability_1976}.
Each regulator minimizes the sum of:
\begin{enumerate*}[label=(\roman*)]
	\item the net emission costs $\Cemm[r,p]\qi[r,p]$ generated by each producer $p$, and
	\item if $\taxb=1$, the negative net total taxes $\ti[r]\qi[r]$, and
	\item the net cost of energy imports $\piI[r]\qAimpI[r'\to r]$ from any other market $r'$, and
	\item the negative net revenues $\piI[r]\qexp[r]$ from energy exports to other countries.
\end{enumerate*}
The import price $\piI[r]$ is a variable linking the different markets since it depends on the energy available from all the regulators. Equivalently, $\piI[r]$ is the shadow price to the so-called \emph{market-clearing constraint}  $\sum_{r' \in\mathcal{R}}  \qAimpI[r\to r'] = \sum_{r\in \mathcal{R}} \qexp[r]$.
In practice, we implement this constraint by introducing the \emph{invisible hand}, a fictitious player whose optimality conditions enforce the market-clearing constraint.  This implies the energy market is a perfect market (i.e., there is perfect competition), similarly to the markets modeled by several authors \citep{egging2010world,egging2008complementarity,gabriel2010solving,Sankaranarayanan2018,feijoo2018future}. Therefore, in any \MNE, $\piI[r]=\piI[r']$ for any $r,r' \in \mathcal{R}$.
Inside the regulator's $r$ market, each producer $p \in \mathcal{P}^r$ solves the convex quadratic problem
\begin{subequations}
	\begin{align}
		\min_{\qi[r,p]}  \quad \quad & \Cilin[r,p] \qi[r,p] + \frac{1}{2}\Ciquad[r,p] {(\qi[r,p])}^2                     
		+\ti[r]\qi[r,p]
		- 
		\left (
		\DemInt[r] - \DemSlope[r] \left(
		\sum_{p'\in\mathcal{P}^r} \qi[r,p'] + \qimp[r] - \qexp[r]
		\right)
		\right )\qi[r,p]\\
		\text{s.t.}\quad       & \qi[r,p] \quad \geq\quad   0, \quad \qi[r,p] \quad \leq \quad \qiCap[r,p]. \qquad 
	\end{align}\label{eq:EC:follower}
\end{subequations}
Each producer $p$ minimizes the sum of
\begin{enumerate*}[label=(\roman*)]
	\item a linear and a quadratic cost of production $\Cilin[r,p] \qi[r,p]+\frac{1}{2}\Ciquad[r,p] {(\qi[r,p])}^2 
	$, and
	\item the tax expenses $\ti[r]\qi[r,p]$, and
	\item the negative net profits given by $\qi[r,p]$ times the current energy price depending on the amount of energy produced inside $r$'s market.
\end{enumerate*}
The producer also has two constraints imposing a non-negative amount of energy $\qi[r,p]$ produced and a capacity cap set at $\qiCap[r,p]$. When $\Ciquad[r,p] \ge 0$, $p$ solves a strictly convex quadratic optimization problem, and therefore, the followers' game in $r$ has a unique equilibrium.

\subsection{Data Generation. } On the one hand, we generate $2$ sets of instances (InstanceSet \emph{A} and \emph{B}) to compare the performances of our algorithms extensively. \emph{InstanceSet A} contains $150$ instances with $3 \le |\mathcal{R}|\le 5$ and $1 \le |\mathcal{P}^r| \le 3$, while \emph{InstanceSet B} contains $50$ instances with $|\mathcal{R}|=7$ and $1 \le |\mathcal{P}^r| \le 3$.
On the other hand, to derive managerial insights and prescriptive recommendations, we generate a real-world Chile-Argentina instance based on real data and a set of instances \emph{InstanceSet Insights} with $|\mathcal{R}|=2$. We employ $3$ realistic taxation schemes for each regulator's $r$ market: 
\begin{enumerate*}[label=(\roman*)]
	\item \emph{Single-Taxation}, where each  $p \in \mathcal{P}^r$ incurs in the same $\ti[r]$ tax,
	\item \emph{Standard Taxation}, where $r$ imposes a custom tax $\ti[r,p]$ on each producer $p \in \mathcal{P}^r$,
	\item \emph{Carbon-Taxation}, where $\ti[r]$ is proportional to $\Cemm[r,p]$ for any $p$, i.e., $\ti[r] = \Cemm[r,p] \textcolor{blue}{\bf t_{e}}^{r}$ where $\textcolor{blue}{\bf t_{e}}^{r}$ is the per-unit emission tax.  
\end{enumerate*}

For any regulator $r$, we randomly draw each $ p \in \mathcal{P}^r$ from $3$ categories: highly polluting producers (e.g., coal and oil plants) with $\Cemm[r,p] \in [300,500]$, averagely polluting producers (e.g., gas-powered plants) with $\Cemm[r,p] \in [100,200]$, and green producers (e.g., solar and wind farms, hydroelectric stations) with $\Cemm[r,p] \in [25,50]$. The associated emission costs $\Cemm[r,p]$ are $USD$ values per GWh of unit-energy and assume that each tonne of CO$_2$ has a social cost of $25\$$. Finally,  $\Cilin[r,p]$, $\Ciquad[r,p]$, and $\qiCap[r,p]$ are in the ranges $[150,300]$, $[0,0.6]$, and $[50, 20000]$, respectively. We refer the reader to the electronic companion for a detailed review of the parameters.

\subsection{Managerial Insights}
\label{sec:sub:strategic}
In our first set of experiments, we attempt to answer the following two managerial questions:

\begin{enumerate}[label=(\roman*)]
	\item [(i)] \textbf{Tax policy.} {Are regulators further reducing their emissions if they consider the carbon tax as a source of income?}
	\item [(ii)] \textbf{Trade policy.} {How does competitive energy trade among different markets affect the overall level of emissions?}
\end{enumerate}
We employ \emph{InstanceSet Insights} with $4$ combinations of parameters for our energy model. First, we employ either a Carbon-Taxation scheme with the revenue term in every regulator's objective or no taxation. Second, we either allow regulators to trade energy or not. 
 
\paragraph{Tax Policy. } 
Several authors argue that carbon tax revenues can further help reduce carbon emissions, promote greener technologies such as carbon sequestration and electric vehicles, and even cover governmental expenses \citep{olsen2018optimal,liu2015economic,amdur2014public}.
Intuitively, the carbon taxes collected by regulators maximizing such incomes may help promote greener energy sources. However, throughout our experiments, we observe the opposite effect. When a regulator maximizes its carbon-tax revenues with $\taxb=1$, it also systematically imposes a smaller carbon tax compared to the case where it does not optimize for carbon-tax revenues ($\taxb=0$). Therefore, as a feedback effect, coal and gas producers generate more energy than they would produce with a greater taxation level, and the regulator collects more tax incomes. While the overall emission levels are modest compared to a no-taxation scheme, they are more significant compared to $\taxb=0$.
Specifically, in $40$ out of the $50$ test instances, both markets' net emissions increased by an average of $13.5$\% with $\taxb=1$. Further, a statistical t-test rejects the null hypothesis ($p$-value of $0.00018$) that the global emissions are equal with and without the regulators maximizing a carbon tax. Similarly, we observe decreased level of energy trade in $30$ out of $50$ instances and, on average, a decrease of about $7.8$\% with $\taxb=1$ compared with $\taxb=0$. However, a similar t-test does not support the hypothesis that the traded quantities of energy in the two cases have the same population mean ($p$-value of $0.29$).

\paragraph{ Trade policy. } We observe a decreased taxation level when the regulators' markets can exchange energy. Quantitatively, the average tax rate drops by $12.9$\% when the markets can trade energy. However, in the $63\%$ of the tests, the tax rate slightly increases when markets can trade. Nevertheless, in the instances with a lower tax rate and trade, the tax rate is significantly lower than in the cases where markets do not trade.
We further study whether an increased trade intensity could exhibit an increased emission level as an externality. However, the overall level of emissions consistently decreases when markets trade energy. Specifically, we observe a substitution effect where clean energy producers fulfill the demand of the emission-intensive market through energy exports. Indeed, the average net emissions drop by $35.9$\% when markets exchange energy, and the net emissions never increase. With the same energy consumption levels, the single market's emissions may increase while the overall emissions in both markets decrease.

\paragraph{General Remarks. }Overall, regulators are not necessarily reducing their net emissions if they consider carbon taxes as a source of revenue, and the free and competitive trade of energy among markets tends to reduce the total emissions. We remark that our insights are sensitive to the producers' cost, capacity, emission factors, and domestic energy demands. Nevertheless, \NASPs provide an effective and flexible framework to perform such analyses and potentially extend them to richer energy models.  We refer the reader to the electronic companion for the expanded set of computational results involving \emph{InstanceSet Insights}.

\subsection{The Chilean-Argentinan Case Study}
\label{sec:sub:chilearg}
We model the Chilean and Argentinean electricity markets through the energy \NASP by using real data from the years $2018-2019$. \footnote{We source the technical and trade data (e.g., fuel consumption, capacity factors, and variable costs as well as the units of energy exchanged) from the Chilean \emph{Comision Nacional de Energia} and the US Energy Information Administration. Additional data is available at \url{https://www.iea.org/countries/chile} and \url{https://www.iea.org/countries/argentina}} 
The two countries started trading electricity in $2016$, with Chile exporting a small amount of electricity ($1558$ MWh) to Argentina.  Furthermore, since the countries signed a cooperation agreement for electricity and gas in $2019$, experts expect an increased level of trade in the future \citep{simsek_review_2019,fondazione_enel_2019}. 
Furthermore, both countries joined the \emph{Paris agreement}, committing to the decarbonization of their energy systems. In this respect, Chile was the first country in Latin America to implement a carbon tax ($5$ USD$\$$ per tonne of $CO_2$), followed by Argentina ($10$ USD$\$$ per tonne of $CO_2$). 
We analyze the impact of an integrated energy market between the two domestic markets of Chile and Argentina. The energy regulators of both countries can impose carbon policies in the form of taxes. We model several different energy producers in each country. In the Chilean case, we consider hydroelectric, solar, wind, natural gas, and coal as electricity sources.  In the Argentinean case, we mainly consider gas-powered thermal and hydroelectric plants. The historical demand is $129$ TWh/year for Argentina and $60$ TWh/year for Chile.
We analyze how the markets react  (i.) under different electricity-trade policies and (ii.) under the forecasted future levels of renewable sources of electricity.

\paragraph{Insights Without Renewable Expansion. } We test our model when the available capacity for renewables is as of $2019$. On the one hand, when the markets cannot trade energy, the overall production tends to favor non-renewable sources. Specifically, $71$\% of the generation in Argentina comes from gas-powered power plants, while hydroelectric plants fulfill the remaining demand. In the Chilean market, coal and gas power plants meet $42$\% of the demand, hydroelectric energy accounts for the $36$\%, and renewable sources (solar and wind) account for approximately the $15$\%. 
On the other hand, when markets trade, we observe a remarkable substitution effect where Chilean imports from Argentina replace carbon-intensive (coal and gas) sources in Chile. Further, we observe an increase in the carbon tax in Chile and, symmetrically, a decrease in Argentina. However, in Argentina, the exports to Chile tend to increase the domestic electricity price, thus contracting the local demand.
Overall, without an increase in the renewable capacity or a significant decrease in carbon's social cost, our model predicts that Argentina's market could experience increased energy prices. Therefore, unless high-capacity renewable sources can operate in the countries, the trade between Chile and Argentina may be limited.

\paragraph{Insights With Renewable Expansion. } To assess the likelihood of future trade under large renewable deployments, we consider two expansions of wind and solar capacities in Chile for $20$ GWh and $40$ GWh, respectively \citep{amigo2021two}. When markets do not trade, we observe an increasing energy price drop and demand increase in Chile; In the $40$ GWh case, the energy price falls by $13$\%, and the domestic demand increases by $20$\%. 
When trade is allowed, the Argentinean market starts importing small quantities of energy from Chile with an expansion of $20$ GWh. With an expansion of $40$ GWh, the Argentinean market intensively imports energy from Chile, up to $12$ TWh/year. Remarkably, similarly to what was observed in Chile, the Argentinean market experiences a price drop and an increase in domestic demand.

\subsection{Performance Analysis}
\label{sec:sub:speed}

In terms of performance analysis, we test \cref{Alg:FullEnumeration} and \cref{Alg:InnerApproximation} on \emph{InstanceSet A} and \emph{B}.
We mark an instance as \emph{solved} when the algorithm either finds an \MNE or certifies its non-existence within the time limit of $1800$ seconds.
\cref{tab:SummaryA,tab:SummaryB} summarize the computational results for \emph{InstanceSetA} and \emph{InstanceSetB}, respectively. 
The upper parts of the tables report results for \cref{Alg:FullEnumeration} (\emph{FE}) and \cref{Alg:InnerApproximation} (\emph{InnerApp}) for generic \MNEs, while the bottom parts report the results for the \PNE variant of \cref{Alg:FullEnumeration} ($FE-P$).
In \cref{Alg:InnerApproximation:extension} of \cref{Alg:InnerApproximation}, we select a total of $\hat{\theta}$ complementarity polyhedra $\tilde \theta^i$ by employing $3$ strategies ($ES$): given a lexicographic order of each player's polyhedra, we add $\hat{\theta}$ polyhedra \emph{sequentially} ($Seq$), \emph{reverse-sequentially} ($RSeq$), or \emph{randomly} ($Rand$). 
In the $Time (s)$ columns, we report the average time in seconds when the algorithm \begin{enumerate*}[label=(\roman*)]
\item finds an \MNE (\emph{EQ}), and
\item proves no \MNE exists (\emph{NO}), and
\item solves the instance or hits the time limit.
\end{enumerate*} 
In the $Wins$ columns, we report the number of times the algorithm outperforms the others in computing time when either an \MNE exists ($EQ$) or not ($NO$). Finally, we report the number of solved instances in the last column ($Solved$).
\emph{InnerApp} outperforms \emph{FE}, being on average twice as fast, and up to $30$ times faster when an \MNE exists (e.g., \emph{InnerApp}-RevSeq-1). Furthermore, in \emph{InstanceSet B}, \emph{InnerApp} manages to solve $47$ out of the $50$ hard instances compared to the $20$ solved by \emph{FE}. When no \MNE exists, \emph{FE} tends to outperform \emph{InnerApp} since the latter necessarily needs to converge to the exact approximation to provide a certificate of non-existence. We remark that both \emph{InnerApp} and \emph{FE} may also return a \PNE, while \emph{FE-P} will necessarily return a \PNE if one exists. Empirically, a \PNE exists in $37.6\%$ and $30.4\%$ of instances in the \emph{InstanceSetA} and \emph{InstanceSetB}, respectively. 

\begin{table}[!ht]
\centering
\caption{Results for \emph{InstanceSetA}. } \label{tab:SummaryA}
\scriptsize

\begin{tabular}{@{}c@{\hspace{2em}}r@{\hspace{1.2em}}r@{\hspace{1.2em}}c@{\hspace{4.0em}}r@{\hspace{2.0em}}r@{\hspace{2.0em}}r@{\hspace{4.0em}}r@{\hspace{2.0em}}r@{\hspace{2.0em}}r@{\hspace{2.0em}}}
	\toprule
	\multicolumn{1}{l}{}   & 
	\multicolumn{3}{c}{\textbf{\textbf{}}}   & 
	\multicolumn{3}{c}{\textbf{Time (s)}} & 
	\multicolumn{2}{c}{\textbf{Wins}} & 
	\multicolumn{1}{l}{}  \\ 
	
	\multicolumn{1}{l}{}   & {\textbf{Algorithm}}  & { \textbf{ES}} & { $\hat{\theta}$} & { \textbf{EQ}} & { \textbf{NO}} & { \textbf{All}} & { \textbf{EQ}} & { \textbf{NO}} & { \textbf{Solved}} \\ \midrule
 & \emph{FE}   & -  & -     & 26.78  & 0.12 & 120.21 &  6 & 82 & 140/149  \\ \midrule
 &  & Seq & 1       & 6.18  & 0.35  & 51.33 & 3  &    0 &   145/149 \\
 &  & Seq & 3       & 16.20 & 0.18  & 55.82 & 5  &    0 &   145/149 \\
 &  & Seq & 5       & 5.85  & 0.15  & 51.08 & 3  &    0 &   145/149  \\
 &  & RSeq & 1   & 7.33  & 0.36  & 3.73 &  26 &    0 &   149/149  \\
 &  & RSeq & 3   & 10.31 & 0.18  & 53.12 & 4  &    0 &   145/149  \\
 &  & RSeq & 5   & 8.68  & 0.15  & 76.41 & 5  &    0 &   143/149  \\
 &  & Rand  & 1   & 4.80  & 0.36  & 26.60 & 8  &    0 &   147/149 \\
 &  & Rand  & 3   & 29.49 & 0.18  & 85.65 & 5  &    0 &   143/149 \\
\multirow{-10}{*}{{ \MNE}} & \multirow{-9}{*}{\emph{InnerApp}} 
 & Rand  & 5   & 21.59 & 0.15  & 58.26 & 2  &    0 &   145/149   \\ \midrule
 \PNE & \emph{FE-P} & - & -  & 6.46   & 0.12 & 328.23  & -- & --  &   122/149 \\
	& 
\end{tabular}

\end{table}

\begin{table}[!ht]
\centering
\caption{Results for \emph{InstanceSetB}.} \label{tab:SummaryB}
\scriptsize
\begin{tabular}{@{}c@{\hspace{2em}}r@{\hspace{1.2em}}r@{\hspace{1.2em}}c@{\hspace{4.0em}}r@{\hspace{2.0em}}r@{\hspace{2.0em}}r@{\hspace{4.0em}}r@{\hspace{2.0em}}r@{\hspace{2.0em}}r@{\hspace{2.0em}}}
	\toprule
	\multicolumn{1}{l}{}   & 
	\multicolumn{3}{c}{\textbf{\textbf{}}}   & 
	\multicolumn{3}{c}{\textbf{Time (s)}} & 
	\multicolumn{2}{c}{\textbf{Wins}} & 
	\multicolumn{1}{l}{}  \\ 
	
	\multicolumn{1}{l}{}   & {\textbf{Algorithm}}  & { \textbf{ES}} & { $\hat{\theta}$} & { \textbf{EQ}} & { \textbf{NO}} & { \textbf{All}} & { \textbf{EQ}} & { \textbf{NO}} & { \textbf{Solved}} \\ \midrule
	
 & \emph{FE}   & -  & -     & 260.29  & 1.12 & 1174.32 &  0 & 2 & 20/50   \\ \midrule
 &  & Seq & 1       & 39.26  & 9.64 & 672.24 &   1  &    0 &   32/50 \\
 &  & Seq & 3       & 62.66 & 3.88  & 616.25 &   1  &    0 &   34/50 \\
 &  & Seq & 5       & 24.03  & 2.83 & 733.97 &   1  &    0 &   30/50  \\
 &  & RSeq & 1   & 171.47 & 9.66 & 262.74 &   27 &    0 &   47/50  \\
 &  & RSeq & 3   & 13.85 & 3.86  & 585.27 &   4  &    0 &   34/50  \\
 &  & RSeq & 5   & 78.57  & 2.83 & 798.90 &   6  &    0 &   29/50  \\
 &  & Rand  & 1   & 34.65 & 9.65  & 497.06 &   0  &    0 &   37/50 \\
 &  & Rand  & 3   & 123.02 & 3.87 & 588.03 &   2  &    0 &   36/50 \\
\multirow{-10}{*}{{ \MNE}} & \multirow{-9}{*}{\emph{InnerApp}} 
 & Rand  & 5   & 39.18 & 2.86  & 711.77 & 4  &    0 &   41/50   \\ \midrule
 
 \PNE & \emph{FE-P} & - & -  & 7.36   & 1.12 & 1441.95  & -- & --  &   10/50 \\
	&
\end{tabular}
\end{table}
 
In order to provide a baseline result, we also compare our algorithm with the $SGM$ algorithm from \citet{Carvalho2020computing}. The $SGM$ algorithm computes equilibria for bounded Integer Programming Games, i.e., simultaneous non-cooperative games where each player solves a bounded parametrized mixed-integer program. We reformulate each Stackelberg game $S^i$ as a parametrized mixed-integer program. Specifically, we reformulate \cref{eq:BilevelSet} as a set of linear inequalities and integer variables associated with the complementarities. Although this reformulation is exact, the resulting parametrized mixed-integer programs may be unbounded and it might prevent $SGM$ from terminating. We implement the market clearing constraints by introducing a virtual player solving the problem $\max_{\pi} \{ \sum_{r \in \mathcal{R}}(\sum_{r' \in\mathcal{R}}  \qAimpI[r\to r'] - \sum_{r\in \mathcal{R}} \qexp[r]) \pi \}$, where $\pi$ is the clearing price. In any \MNE, the virtual player guarantees that $\piI[r] = \pi$ and $\qAimpI[r\to r'] =\sum_{r\in \mathcal{R}} \qexp[r]$ for any $r \in \mathcal{R}$. Since the virtual player solves an unbounded problem, at each iteration, we compute the \MNE among the non-virtual players by fixing the price $\piI[r]$ to a specific value. Once the $SGM$ computes a candidate \MNE with $\qAimpI[r\to r'] \neq \sum_{r\in \mathcal{R}} \qexp[r]$, we increase the price $\pi$ whenever $\qAimpI[r\to r'] > \sum_{r\in \mathcal{R}} \qexp[r]$ and decrease it whenever $\qAimpI[r\to r'] < \sum_{r\in \mathcal{R}} \qexp[r]$.
In \cref{tab:SGM}, we compare the computational results of, for instance, our algorithm $FE$ with the ones of $SGM$ on \emph{InstanceSet A}. Our $FE$ dominates the $SGM$ in computing times (by a factor of $10$) and solved instances. $SGM$ solves less than $37\%$ of the instances, and, when finding an \MNE, the algorithm requires, on average, $308.68$ seconds, an increase of $11.5$ times when compared to $FE$. Similarly, when $SGM$ does not find an \MNE, the algorithm requires $55.44$ seconds on average, compared to $0.12$ seconds required by $FE$. We remark that $SGM$ cannot provide proof of non-existence; thus, it may not terminate when the players' optimization problems are unbounded, and an \MNE does not exist. Indeed, most of the time limits are related to the non-existence of an equilibrium.

\begin{table}[!ht]
	\centering
	\caption{Comparison with \emph{SGM} on \emph{InstanceSet A}. } \label{tab:SGM}
	\scriptsize
	
	\begin{tabular}{l@{\hspace{2em}}r@{\hspace{4em}}r@{\hspace{2em}}r@{\hspace{2em}}r@{\hspace{4em}}r@{\hspace{2em}}r}
		\toprule
		\multicolumn{2}{l}{}   & 
		\multicolumn{3}{c}{\textbf{Time (s)}} & 
		\multicolumn{1}{l}{}  \\ 
			
		\multicolumn{1}{l}{} & {\textbf{Algorithm}} & { \textbf{EQ}} & { \textbf{NO}} & { \textbf{All}} & { \textbf{Solved}} \\ \midrule
		                     & \emph{SGM}           & 308.68         & 55.44          & 1191.92         & 55/149             \\
		                     & \emph{FE}            & 26.78          & 0.12           & 120.21          & 140/149            \\ \bottomrule
	\end{tabular}
	
\end{table}

\section{Concluding Remarks}\label{sec:conclusions}
We introduced \NASPs, a class of non-cooperative and simultaneous games among Stackelberg games. \NASPs capture complex hierarchical interactions among decision-makers solving optimization problems, and they can express the complexity associated with many economic markets. We employed theory and tools from optimization -- specifically, polyhedral theory -- and we provided a series of theoretical and algorithmic characterizations of \NASPs.
From a theoretical perspective, we proved that the problem of deciding if a \NASP admits an \MNE is \SigmaTwoP. Furthermore, we demonstrated the equivalence of computing an \MNE in a \NASP and computing a \PNE in a convexified version of the game. From an algorithmic standpoint, we provided an exact algorithm to compute and select \MNEs, and a variant for computing \PNEs. We further introduced a refined version of our original algorithm to exploit an increasingly-accurate sequence of convex inner approximations of the game. 
From a practical perspective, we contextualized \NASPs in international energy markets by proposing a data-rich and realistic model capable of providing valuable managerial insights from Nash equilibria. We analyzed a real-world case study of the Chilean-Argentinean energy market and unveiled counterintuitive consequences of policymaking for climate change-aware regulators.

We believe this paper establishes a solid benchmark for future work combining optimization and game theory. Moreover, it exposes the need for novel game theory frameworks capturing the complex interactions among self-driven decision-makers, as extensively motivated in our energy applications. 
On the one hand, we hopefully foresee further methodological developments extending our approach to other hierarchical games (e.g., multi-leader games with common interacting followers) and novel applications of \NASPs. On the other hand, we hope our analysis of energy markets could be extended to other domains to derive insightful managerial insights.

\newpage
\bibliography{Resources/Main}

\newpage
\section*{Supplementary Material}
In this electronic companion, we complement the proofs of \cref{sec:Hardness} in \cref{sec_sup:proofs}, and we provide an instance of \NASP without \PNEs in \cref{sec_sup:noPNE}. Finally, in \cref{sec_sup:instances}, we provide a detailed overview of the computational tests.

\section{Proof of \cref{thm:MixHard}}
\label{sec_sup:proofs}

Before proving \cref{thm:MixHard}, we formally define the concept of \emph{trivial \NASP} and introduce the two technical lemmata regarding the extended formulations of the union of polyhedra.

\subsection{Trivial \NASP and Extended Formulations, and Stackelberg Games}
\begin{definition}[Trivial NASP]\label{Def:TrivNASP}	A trivial \NASP $N=(S^1,S^2$) is a \NASP with $n=2$ where, for any $i$, $P^i$ is a simple Stackelberg game and each leader has a single follower solving a linear program with a simple parameterization with respect to its leader's variables.
\end{definition}
The additional assumptions of a trivial \NASP are seemingly strong. Specifically, we limit each leader to having one follower and each follower to solve a linear program with a simple parameterization with respect to its leader's variables.  
\citet{Basu2018a} proved that any finite union of polyhedra is the feasible region of a simple Stackelberg game in a lifted space. We formalize their result for $2$ polyhedra in \cref{lem:union}.

\begin{lemma}\label{lem:union}
  Consider the union of two polyhedra $\mathcal S := \left\{ (h,y,w)\in\R^3_+: h=w,\;  y =1 \right\} \cup \left\{ (h,y,w) \in \R_+^3: h=0,\;y=0 \right\}$.
  $\mathcal S$ has an extended formulation as a feasible set of a simple Stackelberg game.
\end{lemma}

\begin{proof}[Proof of \cref{lem:union}.]
Let $z_1,z_2,\dots$ be the variables in the lifted space, i.e., variables that can be projected out. The extended formulation of $\mathcal{S}$ is equivalent to the Stackelberg feasible region
  \begin{subequations}
  \begin{align}
w \quad \geq\quad 0, \quad  y \quad\geq\quad 0 \quad  h \quad\geq\quad 0,  & \quad y \quad\leq\quad 1,\quad h \quad\leq\quad w, \quad z_1,\dots,z_6 \quad\geq\quad 0\\
(z_1,,\dots,z_6) \quad&\in\quad \arg \min_z \left\{ 
  \sum_{i=1}^6 z_i: \begin{array}{ccc}
	z_1 \geq h-w &;& z_1 \geq -h\\
	z_2 \geq 1-y &;& z_2 \geq -h\\
	z_3 \geq y-1 &;& z_3 \geq -h\\
	z_4 \geq w-h &;& z_4 \geq -y\\
	z_5 \geq h-w &;& z_5 \geq -y\\
	z_6 \geq y-1 &;& z_6 \geq -y 
  \end{array}
\right\}.
  \end{align} 
	\label{eq:hardLem}
  \end{subequations}
The above formulation is the feasible set of a Stackelberg game. 
\end{proof}

Furthermore, we show that if two sets $S\subseteq\R^{n_1}$ and $T\subseteq\R^{n_2}$ have an extended formulation as Stackelberg feasible regions, so does $S\times T$.

\begin{lemma}\label{lem:BilCart}
  Suppose $S\subseteq\R^{n_1}$ and $T\subseteq\R^{n_2}$ have an extended formulation as Stackelberg feasible regions. Then, $S\times T$ has an extended formulation as a Stackelberg feasible region.
\end{lemma}

\begin{proof}[Proof of \cref{lem:BilCart}.]
Let the extended formulation of $S$ be $\{ (w,y): A_Sw+B_Sy \leq b_S, \; y\in\arg\min\{ f_S^Ty:C_Sw+D_Sy \leq g_S \} \}$, and the extended formulation of $T$ be $\{ (w,y): A_Tw+B_Ty \leq b_T, \; y\in\arg\min\{ f_T^Ty:C_Tw+D_Ty \leq g_T \} \}$. The extended formulation of $S\times T$ is
\begin{align*}
  A_Sw+B_Sy \quad \leq \quad b_S, & \quad A_Tu + B_Tv \quad \leq \quad b_T, \\
  (y,v) \quad& \in \quad  \arg\min \left\{ f_S^Ty+f_T^Tv: \begin{array}{c} C_Sw+D_Sy \quad \leq \quad g_S \\ C_Tu+D_Ty \quad \leq \quad g_T \end{array} \right\}.  
\end{align*}
\end{proof}

\subsection{Proof of \cref{thm:MixHard}}
We prove the hardness result by performing a reduction of a trivial \NASP to the \SSI problem (\cref{Def:SubSum}).

\begin{proof}[Proof of \cref{thm:MixHard}.]
  We reduce \SSI to the problem of deciding the existence of an \MNE for the trivial \NASP in \cref{eq:MixHard}. 
  Let $Q = \sum_{i=1}^k q_i$, and let the two Stackelberg games $(S^1,S^2)$ be the Latin and the Greek game, respectively. The proof develops around the following $7$ core claims.

\begin{claim}\label{cl:MNE:NASP}
  The game defined in \cref{eq:MixHard} is a trivial \NASP.
\end{claim}
  
\begin{claim}\label{cl:MNE:Disj}
  The region of space for $w$ defined by \cref{eq:MLat:xpos,eq:MLat:bil} is the Cartesian product of $(\{w_i:w_i\leq 0\}\cup \{w_i:w_i\geq 1\}  )$ for $i=0,\dots,k$. Similarly, the region of space for $\xi$ defined by \cref{eq:Gr:bil,eq:Gr:chipos} is the Cartesian product of $(\{\xi_i:\xi_i\leq 0\}\cup \{\xi_i:\xi_i\geq 0\} )$ for any $i=1,\dots,k$.
\end{claim}

\begin{claim}\label{cl:MNE:xInt}
  $w_{k+3r+1}$ takes only integer values.
\end{claim}

\begin{claim} \label{cl:MNE:squ}
  $ (w_{k+3r+1})^2 = \sum_{i=1}^r 2^{i-1}w_{k+i} + pw_{k+3r+1}$ holds for the Latin player's feasible set. 
\end{claim}

\begin{claim}\label{cl:MNE:LatCons}
  Given some $\xi_{r+1} \in \Z$ between $p$ and $t-1$, the Latin player has a profitable unilateral deviation for any feasible strategy with $w_{k+3r+1} \neq \xi_{r+1}$.
\end{claim}
\begin{claim}\label{cl:MNE:YES}
  If \SSI has decision YES, then \cref{eq:MixHard} has a \PNE (and hence an \MNE).
\end{claim}

\begin{claim}\label{cl:MNE:NO}
  If \SSI has decision NO, then \cref{eq:MixHard} has no \MNE.
\end{claim}

\noindent {\em Proof of \cref{cl:MNE:NASP}.}  
The constraints and the players' objectives are linear in their variables. The constraints \cref{eq:MLat:S} are valid due to \cref{lem:union}. Because of \cref{lem:BilCart}, we can reformulate the constraints \cref{eq:MLat:S,eq:MLat:bil} as the feasible region of a Stackelberg game. Furthermore, each follower solves a linear program parametrized in its leader's variables. Therefore, \cref{eq:MixHard} is a trivial \NASP.

\noindent {\em Proof of \cref{cl:MNE:Disj}.}  
The constraints in \cref{eq:MLat:bil} enforce that $y_i \geq \max (-w_i, w_i-1)$. Since $y_i$ is minimized, it necessarily equals $\max(w_i-1, -w_i)$. However, as of \cref{eq:MLat:xpos}, $y_i$ is non-negative. Therefore, either $w_i \leq 0$ or $1-w_i \leq 0$. Equivalently, we can express this disjunction as $(\{w_i:w_i\leq 0\}\cup \{w_i:w_i\geq 1\}  )$. The same reasoning holds for the Greek player.

\noindent {\em Proof of \cref{cl:MNE:xInt}.} From \cref{eq:MLat:S}, each $w_{k+r+i}$ is either $0$ or $1$ for any $i=1,\dots,r$. Depending on the value of $w_{k+r+i}$, the variable is in either one of the polyhedra defining $S$ (see \cref{lem:union}). Moreover, since in \cref{eq:MLat:xbin} the RHS is a sum of integers, the LHS $w_{k+3r+1}$ is also an integer.

\noindent {\em Proof of \cref{cl:MNE:squ}.} Consider the set $S$ defined in \cref{lem:union}. For any point $h = w$ and $y = 1$ in the first polyhedron, $h = wy$. Similarly, for any point $h = 0$ and $y = 0$ in the second polyhedron, $h = wy$. Thus, the nonlinear equation $h = wy$ is valid for the set $S$. Similarly, by multiplying both sides of \cref{eq:MLat:xbin} with $w_{k+3r+1}$, we obtain
\begin{align*}
(w_{k+3r+1})^2 \quad&=\quad  pw_{k+2r+1} + \sum_{i=1}^r 2^{i-1} w_{k+r+i}w_{k+3r+1} \\
\quad&=\quad pw_{k+3r+1} + \sum_{i=1}^r 2^{i-1} w_{k+r+i}w_{k+2r+i}\\
\quad&=\quad pw_{k+3r+1} + \sum_{i=1}^r 2^{i-1} w_{k+i}.
\end{align*}
The second equality follows from \cref{eq:MLat:xeq}, i.e., $w_{k+3r+1}=w_{k+2r+i}$.

\noindent {\em Proof of \cref{cl:MNE:LatCons}.} If $\xi_{r+1} \in [p,t-1]$, then $w_{k+3r+1} = \xi_{r+1}$ is feasible for the Latin player. Consider the last two terms of the Latin player's objective function. From \cref{cl:MNE:squ}, we can rewrite them as $(Q+1) (2\xi_{r+1} w_{k+3r+1} - w_{k+3r+1}^2)$. The latter expression attains its maximum value for $w_{k+3r+1} = \xi_{r+1}$. We argue that the Latin player can never be optimal by choosing $w_{k+3r+1} \neq \xi_{r+1}$. Note that $w_{k+3r+1}$ only takes integer values (\cref{cl:MNE:xInt}), and, for any choice $w_{k+3r+1} \neq \xi_{r+1}$, the objective function decreases at least of $Q+1$. However, even if each of the other objective's terms attains its maximum possible value, the largest value they add up to is $0.5 + Q$. Since $0.5 + Q < Q+1$, the claim follows.

\noindent {\em Proof of \cref{cl:MNE:YES}.} Let $s$ be an integer so that $p \leq s < t$ and for all $I\subseteq \left\{ 1,\dots, k \right\}$ with $\sum_{i\in I}q_i \neq s$, and let $b_1,\dots,b_r \in \left\{ 0,1 \right\}$ be the unique $r$-bit binary representation of $s-p$. Consider the pure strategies 
  \begin{align}
	w_{k+3r+1}  \quad&=\quad  s, \nonumber \\
	w_{k+2r+i}  \quad&=\quad  s & i=1,\dots,r, \nonumber \\
	w_{k+r+i}  \quad&=\quad b_i  & i=1,\dots,r,\nonumber  \\
	w_{k+i} \quad&=\quad b_is & i=1,\dots,r, \nonumber \\
	w_0 \quad&=\quad 1, \nonumber \\ %
	\xi_0  \quad&=\quad  0, \nonumber \\
	\xi_i \quad&=\quad b_i  & i=1,\dots,r, \nonumber \\
	\xi_{r+1} \quad&=\quad s \label{eq:strategy}.
  \end{align}
Select $w_{i} \in \{0,1\}$ for $i=1,\dots, k$ so that  $\sum_{i=1}^k q_iw_i$ is the largest value not exceeding $s$. Since we assume this is a YES instance of \SSI, $\sum_{i=1}^k q_iw_i\leq s-1$, and the strategy \cref{eq:strategy} is feasible for both the players. On the one hand, the Latin player has no feasible and profitable deviation since (i.) $w_{k+3r+1}$ should take integer values as of \cref{cl:MNE:LatCons}, and (ii.) $w_{k+3r+1}=s$ as of \cref{eq:strategy}, and (iii.) the first two terms in the above strategy already attain the largest possible value not violating \cref{eq:MLat:knap}. 
On the other hand, the Greek player cannot improve its objective value from $0$ since $w_0 = 1$. Therefore, the strategy in \cref{eq:strategy} is a \PNE.

\noindent {\em Proof of \cref{cl:MNE:NO}.} Recall that $w_{k+3r+1}$ is an integer between $p$ and $t-1$. For any choice of $w_{k+3r+1}$, $w_0 = 0$ and $w_1,\dots, w_k$ are so that \cref{eq:MLat:knap} holds at the equality. Therefore, there is no incentive to set $w_0 = 1$ since the choice will only increase the player's objective of $0.5$. However, with $w_0 = 0$, the Greek player can assign arbitrarily large values to $\xi_0$. Hence, there always exists a value of $\xi_0$ constituting a profitable deviation for the Greek player. Thus, the game has no \MNE.
\end{proof}

\section{\NASP with no \PNE but only an \MNE}
\label{sec_sup:noPNE}
\begin{example}
\label{app:PneMne}
Consider the following Latin-Greek trivial \NASP.

\begin{subequations}
\noindent{\bf Latin Player}
  \begin{align}
      \max_{w,y}\quad&  w_1\xi_1 + w_2\xi_2 &\\
      \text{s.t.} \quad &w, y \quad\geq\quad 0,\quad w \quad\leq\quad 1, \quad w_1+w_2\quad=\quad 1,\\
      &y\quad\in\quad \arg\min_y \left\{ y_1+y_2:
	\begin{array}{l}
	y_i \geq -w_i\\
	y_i \geq w_i-1
	\end{array}\,\text{for } i=1,2
	\right \}.
  \end{align}
  \noindent{\bf Greek Player}
  \begin{align}
      \max_{\xi,\chi}\quad&  w_2\xi_1 + w_1\xi_2 &\\
      \text{s.t.} \quad & \xi, \chi \quad\geq\quad 0, \quad \xi \quad\leq\quad 1, \quad \xi_1+ \xi_2 \quad=\quad 1,\\
      & \chi\quad\in\quad \arg\min_\chi \left\{ \chi_1+\chi_2:
	\begin{array}{l}
	\chi_i \geq -\xi_i\\
	\chi_i \geq \xi_i-1
	\end{array}\,\text{for } i=1,2
	\right \}.
  \end{align}\label{eq:PneMne}
\end{subequations}
The only feasible strategies are $\{(1,0,0,0), (0,1,0,0)\}$ for both $(w_1,w_2,y_1,y_2)$ and $(\xi_1,\xi_2,\chi_1,\chi_2)$. We can equivalently reformulate the game as a normal-form game as follows. If the Latin and the Greek player select the same strategy, then the Latin player gets a payoff of $1$, and the Greek player gets a payoff of $0$. Otherwise, if they select different strategies, the Latin player gets a payoff of $0$, and the Greek player gets a payoff of $1$.  The only Nash equilibrium of this normal-form game is an \MNE, and, therefore, no \PNE exists.
\end{example}

\section{Computational Tests}
\label{sec_sup:instances}

\subsection{Instance sets}
Our instances and code are available in our Github repository at \url{https://github.com/ssriram1992/EPECsolve}. A generalization of our code is also in \url{https://github.com/ds4dm/ZERO}. We tested our algorithms on an $8$-cores Intel(R) Xeon Gold 6142 with $32$GB of RAM and $Gurobi$ 9.0. We generated three instances sets for our tests:

\begin{enumerate}[label=(\roman*)]
\item \emph{InstanceSet A} contains $149$ instances with $n \in [3,5]$ and up to $3$ producers per regulator.
\item \emph{InstanceSet B} contains $50$ instances with $n=7$ and up to $3$ producers per regulator. These instances cannot be solved by \cref{Alg:FullEnumeration} within $10$ seconds.
\item \emph{InstanceSet Insights} contains $50$ instances with $n=22$ countries and $3$ producers per regulator. 
\end{enumerate} In \cref{tab:Parameters}, we provide a brief overview of the parameters we employed to generate our instances.

\begin{table}[!ht]
\renewcommand{\arraystretch}{1.5}
\scriptsize
\begin{tabularx}{\textwidth}{SlSlX}
\toprule
\textbf{Parameter}                & \textbf{Distribution}                     & \textbf{Notes}                                              \\ \midrule
\textbf{Capacities $\qiCap[r,p]$}      & $50, 100, 130, 170, 200, 1000, 1050, 20000$ & {Each follower’s capacity is randomly drawn from these values independently from the type of production (i.e., renewable vs non-renewable). }                                      \\
\textbf{Emission Costs $\Cemm[r,p]$ }  & $25, 50, 100, 200, 300, 500, 550, 600$      & {The first two values are reserved for green producers. The following two are for averagely-polluting producers, while the remaining three are for highly-polluting ones.}                                          \\
\textbf{Linear Costs $\Cilin[r,p]$}    & $150, 200, 220, 250, 275, 290, 300$         & {Linear costs are generally inversely proportional to the emission cost.  }                    \\
\textbf{Quadratic Costs $\Ciquad[r,p]$} & $0, 0.1, 0.2, 0.3, 0.5, 0.55, 0.6$          & {Quadratic costs are generally inversely proportional to the emission cost. }                         \\
\textbf{Tax Caps $\tiCap[r]$}        & $0, 50, 100, 150, 200, 250, 275, 300$       & {Tax caps are generally inversely proportional to the emission cost.}                                   \\
\textbf{Demand Alpha $\DemInt[r]$}    & $275, 300, 325, 350, 375, 450   $           & {Each market's alpha is randomly drawn from this set. }                  \\
\textbf{Demand Beta $\DemSlope[r]$}     & $0.5, 0.6, 0.7, 0.75, 0.8, 0.9 $            & {Each market's beta is randomly drawn from this set.}                                 \\
\textbf{Price Cap $\piCeil[r]$}       & $0.8, 0.85, 0.90, 0.95   $                  & {Each market's price-limit is randomly drawn from this set. The final price limit is made of the product of this value and $\DemInt[r]$. }                                       \\
\textbf{Tax Paradigm}    & Standard, Single, Carbon                  & {\emph{Single-Taxation}, where each  $p \in \mathcal{P}^r$ incurs in the same $\ti[r]$ tax.
	\emph{Standard Taxation}, where $r$ imposes a custom tax on each producer $p \in \mathcal{P}^r$, i.e., a tax $\ti[r,p]$.
	\emph{Carbon-Taxation}, where $\ti[r]$ is proportional to $\Cemm[r,p]$ for any $p$, i.e., $\ti[r] = \Cemm[r,p] \textcolor{blue}{\bf t_{e}}^{r}$ with $\textcolor{blue}{\bf t_{e}}^{r}$ being the per-unit emission tax.} \\ \bottomrule
\end{tabularx}
\caption{Description of the parameters for our instances.}.
\label{tab:Parameters}
\end{table}

\subsection{Full Results Tables} \cref{tab:MNE:A,tab:MNE:B} contains the full results for \emph{InstanceSetA} and \emph{InstanceSetB}, respectively. The first three columns are the instance number ($\#$), the number of leaders ($n$), and -- for each leader -- their respective number of followers in squared parenthesis ($F$). In the following $10$ columns, we report the computing time associated with the \cref{Alg:FullEnumeration} (\emph{FE}) and the various inner approximation configurations. Specifically, each column identifies the algorithm by its name, and possibly its extension strategy and $\hat{\theta}$ parameter. In the column $MNE$, we report the instance's status, namely, whether at least one algorithm determined an \MNE exists ($YES$) or not ($NO$). Finally, in the last two columns, we report the computing time for the \PNE version of \cref{Alg:FullEnumeration} (\emph{FE-P}) and the status associated with the problem of determining a \PNE.

\cref{tab:Insights} reports the results for \emph{InstanceSet Insights}. We report the numeric values of $0$ with dashes.
The columns are, in order of appearance: 
\begin{enumerate*}[label=(\roman*)]
\item the instance's number ($\#$), and
\item the boolean tax switch (T$^a$), i.e., T$^a$ if taxation is allowed, and
\item the boolean trade switch (T$^r$), and
\item the set of results associated with each regulator's market.
\end{enumerate*}
Specifically, for both \emph{Country One} and \emph{Country Two}, we report:
\begin{enumerate*}[label=(\roman*)]
\item the unit-energy production level (\emph{qTot}), and
\item the domestic price per unit-energy (\emph{\$D}), and
\item the imports (\emph{I}) and exports (\emph{E}), and
\item the export price per unit-energy (\emph{\$E}), and
\item the tax level per unit-energy (\emph{T}).
\end{enumerate*}
For each producer, we report: 
\begin{enumerate*}[label=(\roman*)]
\item the type (\emph{Ty}), namely, $C$ for coal, $G$ for gas, or $S$ for solar, and
\item the emission cost per unit-energy (\emph{EC}), and 
\item the production level ($q$).
\end{enumerate*}

\begin{landscape}
{\scriptsize
\begin{longtable}{ccc@{\hspace{4em}}l@{\hspace{2em}}lll@{\hspace{2em}}lll@{\hspace{2em}}lll@{\hspace{2em}}l@{\hspace{4em}}ll}
\caption{Detailed results for {\em InstanceSet A}.}\\    %
\toprule
\textbf{\#} & \textbf{n} & \textbf{F}        & \textbf{FE} & \textbf{Seq-1} & \textbf{Seq-3} & \textbf{Seq-5} & \textbf{RSeq-1} & \textbf{RSeq-3} & \textbf{RSeq-5} & \textbf{Rand-1} & \textbf{Rand-3} & \textbf{Rand-5} & \textbf{\MNE} & \textbf{FE-P} & \textbf{\PNE}        \\ \midrule \endhead
1        & 3 & {[} 1 2 2 {]}     & 0.04    & 0.04    & 0.04    & 0.04    & 0.04   & 0.04    & 0.04    & 0.04    & 0.04    & 0.04    & NO & 0.04    & NO \\
2        & 3 & {[} 2 2 3 {]}     & 0.07    & 0.15    & 0.08    & 0.08    & 0.16   & 0.08    & 0.08    & 0.16    & 0.08    & 0.08    & NO & 0.07    & NO \\
3        & 3 & {[} 2 2 2 {]}     & 0.05    & 0.13    & 0.06    & 0.06    & 0.13   & 0.06    & 0.06    & 0.13    & 0.06    & 0.06    & NO & 0.05    & NO \\
4        & 3 & {[} 1 2 1 {]}     & 0.04    & 0.10    & 0.05    & 0.05    & 0.10   & 0.05    & 0.05    & 0.10    & 0.05    & 0.05    & NO & 0.04    & NO \\
5        & 3 & {[} 2 1 3 {]}     & 0.33    & 0.04    & 0.04    & 0.04    & 0.04   & 0.04    & 0.04    & 0.04    & 0.04    & 0.04    & YES     & 0.09    & YES     \\
6        & 3 & {[} 2 2 1 {]}     & 0.06    & 0.18    & 0.11    & 0.07    & 0.18   & 0.11    & 0.07    & 0.18    & 0.11    & 0.07    & NO & 0.06    & NO \\
7        & 3 & {[} 2 2 2 {]}     & 0.06    & 0.14    & 0.07    & 0.07    & 0.14   & 0.07    & 0.07    & 0.14    & 0.07    & 0.07    & NO & 0.06    & NO \\
8        & 3 & {[} 1 2 1 {]}     & 0.10    & 0.08    & 0.08    & 0.08    & 0.08   & 0.08    & 0.08    & 0.08    & 0.08    & 0.08    & YES     & 0.08    & YES     \\
\end{longtable}
}

{\scriptsize
\begin{longtable}{ccc@{\hspace{4em}}l@{\hspace{2em}}lll@{\hspace{2em}}lll@{\hspace{2em}}lll@{\hspace{2em}}l@{\hspace{4em}}ll}
\caption{Detailed results for {\em InstanceSet B}.}\\    %
\toprule
\textbf{\#} & \textbf{n} & \textbf{F}        & \textbf{FE} & \textbf{Seq-1} & \textbf{Seq-3} & \textbf{Seq-5} & \textbf{RSeq-1} & \textbf{RSeq-3} & \textbf{RSeq-5} & \textbf{Rand-1} & \textbf{Rand-3} & \textbf{Rand-5} & \textbf{\MNE} & \textbf{FE-P} & \textbf{\PNE}        \\ \midrule \endhead
0  & 7 & {[} 1 1 2 2 2 2 2 {]} & TL & TL & TL & TL & TL & 8.86    & TL & 240.99  & TL & TL & YES         & TL & TL          \\
1  & 7 & {[} 1 1 1 2 1 2 1 {]} & 62.12   & 2.82    & 1.93    & 102.97  & 1.64    & 6.91    & 279.06  & 1.87    & 5.98    & 83.84   & YES         & TL & TL          \\
2  & 7 & {[} 2 3 2 1 3 2 2 {]} & TL & 6.28    & 6.29    & 6.32    & 6.43    & 6.42    & 6.35    & 6.28    & 6.35    & 6.32    & YES         & TL & TL          \\
3  & 7 & {[} 1 2 2 1 1 3 1 {]} & 572.69  & 0.51    & 0.50    & 0.50    & 0.51    & 0.50    & 0.50    & 0.50    & 0.50    & 0.50    & YES         & TL & TL          \\
4  & 7 & {[} 1 3 2 2 1 2 1 {]} & TL & 1.34    & 1.33    & 1.33    & 1.33    & 1.35    & 1.32    & 1.33    & 1.35    & 1.32    & YES         & TL & TL          \\
5  & 7 & {[} 2 1 2 2 1 2 2 {]} & 192.37  & 0.28    & 0.27    & 0.28    & 0.28    & 0.28    & 0.27    & 0.27    & 0.27    & 0.27    & YES         & 2.40    & YES         \\
6  & 7 & {[} 2 2 2 3 3 2 1 {]} & TL & TL & TL & TL & 20.02   & TL & TL & TL & TL & TL & YES         & TL & TL          \\
7  & 7 & {[} 2 2 1 1 3 2 3 {]} & TL & TL & TL & TL & 464.21  & TL & TL & TL & TL & TL & YES         & TL & TL          \\
\end{longtable}
}

\newpage
{\tiny
\setlength{\tabcolsep}{1pt}
\begin{longtable}{ccc@{\hspace{3em}}llllll@{\hspace{1.5em}}lllllllll@{\hspace{3em}}llllll@{\hspace{1.5em}}lllllllll}
\caption{Instances' solutions for \emph{InstanceSetInsights}. 
} \\    %
\toprule
\textbf{}   & \textbf{}    & \multicolumn{1}{l}{\textbf{}}      & \multicolumn{15}{c}{\scriptsize\textbf{Country One}}                                                                                                                                                                                                           & \multicolumn{15}{c}{\scriptsize\textbf{Country Two}}                                                                                                                                                                                      \\
\textbf{}   & \textbf{}    & \multicolumn{1}{l}{\textbf{}}      & \textbf{}     & \textbf{}       & \textbf{}    & \textbf{}    & \textbf{}       & \textbf{}    & \multicolumn{3}{l}{\textbf{Follower 1}}  & \multicolumn{3}{l}{\textbf{Follower 2}}  & \multicolumn{3}{l}{\textbf{Follower 3}}                       & \textbf{}     & \textbf{}       & \textbf{}    & \textbf{}    & \textbf{}       & \textbf{}    & \multicolumn{3}{l}{\textbf{Follower 1}}  & \multicolumn{3}{l}{\textbf{Follower 2}}  & \multicolumn{3}{l}{\textbf{Follower 3}}  \\
\textbf{\#} & \textbf{T$^a$} & \multicolumn{1}{l}{\textbf{T$^r$}} & \textbf{qTot} & \textbf{\$D} & \textbf{I} & \textbf{E} & \textbf{\$E} & \textbf{T} & \textbf{Ty} & \textbf{EC} & \textbf{q} & \textbf{T} & \textbf{EC} & \textbf{q} & \textbf{Ty} & \textbf{EC} & \multicolumn{1}{l}{\textbf{q}} & \textbf{qTot} & \textbf{\$D} & \textbf{I} & \textbf{E} & \textbf{\$E} & \textbf{T} & \textbf{Ty} & \textbf{EC} & \textbf{q} & \textbf{T} & \textbf{EC} & \textbf{q} & \textbf{Ty} & \textbf{EC} & \textbf{q} \\ \midrule \endhead
0  & 0   & 0     & 83,33       & 300,00 &-   &-   & -         & 80,00  & C          & 500 &-  & G          & 200 & 41,67  & S          & 25  & 41,67  & 60,00       & 255,00 &-   &-   & -         & 0,30   & C          & 500 & 27,76 & G          & 100 & 27,76  & S          & 50  & 4,48   \\
0  & 0   & 1     & 98,81       & 300,00 & 0,56   & 16,03  & 279,43    & 80,00  & C          & 500 &-  & G          & 200 & 57,14  & S          & 25  & 41,67  & 44,52       & 255,00 & 16,03  & 0,56   & 278,43    & 9,35   & C          & 500 & 20,52 & G          & 100 & 20,52  & S          & 50  & 3,49   \\
0  & 1   & 0     & 83,33       & 300,00 &-   &-   & -         & 0,14   & C          & 500 & 7,35  & G          & 200 & 37,22  & S          & 25  & 38,76  & 60,00       & 255,00 &-   &-   & -         &-   & C          & 500 & 27,42 & G          & 100 & 27,88  & S          & 50  & 4,69   \\
0  & 1   & 1     & 119,32      & 300,00 &-   & 35,99  & 103255,86 & 0,07   & C          & 500 & 32,02 & G          & 200 & 47,10  & S          & 25  & 40,20  & 24,01       & 255,00 & 35,99  &-   & 103256,86 & 0,07   & C          & 500 & 0,14  & G          & 100 & 22,43  & S          & 50  & 1,45   \\
1  & 0   & 0     & 37,50       & 270,00 &-   &-   & -         & 11,41  & C          & 300 & 29,69 & G          & 100 & 7,81   & S          & 50  &-   & 50,00       & 315,00 &-   &-   & -         & 15,77  & G          & 100 & 26,92 & S          & 50  & 11,54  & S          & 25  & 11,54  \\
1  & 0   & 1     &-        & 270,00 & 46,59  & 9,09   & 63,54     & 50,00  & C          & 300 &-  & G          & 100 &-   & S          & 50  &-   & 87,50       & 315,00 & 9,09   & 46,59  & 62,54     & 7,27   & G          & 100 & 36,36 & S          & 50  & 25,57  & S          & 25  & 25,57  \\
\end{longtable}
}

\end{landscape}

\end{document}